\documentclass{article}
\usepackage{amsmath,amssymb,amsthm,graphicx}
\usepackage{color}
\usepackage{booktabs}
\usepackage{multirow}
\usepackage{xspace}
\usepackage{paralist}
\usepackage[d]{esvect}
\usepackage{mathtools}

\usepackage{authblk}

\usepackage[color=white]{todonotes}

\usepackage[numbers]{natbib}

\usepackage[margin=0pt]{subcaption}

\usepackage{tikz}
\usetikzlibrary{decorations,arrows,petri,topaths,backgrounds,shapes,positioning,fit,calc,decorations.pathreplacing,patterns,intersections}

\usepackage[colorlinks,citecolor=green!60!black,linkcolor=red!60!black,pdfdisplaydoctitle,
menucolor=orange!40!black,filecolor=magenta!40!black,urlcolor=blue!40!black,linkcolor=red!40!black,citecolor=green!40!black,colorlinks,pagebackref]{hyperref}

\newcommand{\mytabref}[1]{Theorem~\ref{#1}}

\newcommand{\mytabprop}[1]{Proposition~\ref{#1}}

\newcommand{\mytabcor}[1]{Corollary~\ref{#1}}

\newlength{\additionaltextheight}
\newlength{\additionaltextwidth}
\usepackage{cleveref}

\renewcommand*{\backref}[1]{}
\renewcommand*{\backrefalt}[4]{%
\ifcase #1%
\marginpar{\tiny no cite}
\or
 $\rightarrow$~p.~#2.%
\else
  $\rightarrow$~pp.~#2.%
\fi
}
\hypersetup{pdftitle={Distance to restricted domains}}

\crefname{table}{Table}{Tables}

\crefname{section}{Section}{Sections}

\newtheorem{theorem}{Theorem}
\crefname{theorem}{Theorem}{Theorems}

\newtheorem{corollary}{Corollary}
\crefname{corollary}{Corollary}{Corolaries}

\newtheorem{obs}{Observation}
\crefname{obs}{Observation}{Observations}

\newtheorem{prop}{Proposition}
\crefname{prop}{Proposition}{Propositions}

\newtheorem{lemma}{Lemma}
\crefname{lemma}{Lemma}{Lemmas}

\theoremstyle{definition}

\crefname{example}{Example}{Examples}

\crefname{reduction}{Reduction}{Reductions}
\newtheorem{definition}{Definition}
\crefname{definition}{Definition}{Definitions}

\crefname{subsection}{Section}{Sections}

\crefname{figure}{Figure}{Figures}

\newcommand{\seq}[1]{\langle #1\rangle}
\newcommand{\pref}{\succ}
\newcommand{\spref}{{\ensuremath{\scriptstyle \succ }}}

\newcommand{\truevalue}{\textrm{true}}
\newcommand{\falsevalue}{\textrm{false}}

\newcommand{\problemPiVDlong}{\textsc{$\Pi$~Maverick Deletion}\xspace}
\newcommand{\problemPiADlong}{\textsc{$\Pi$~Alternative Deletion}\xspace}

\newcommand{\problemPiVD}{\problemPiVDlong}
\newcommand{\problemPiAD}{\problemPiADlong}

\newcommand{\problemValueVD}{\textsc{Value-Restricted Maverick Deletion}\xspace}

\newcommand{\problemMediumAD}{\textsc{Me\-di\-um-Re\-strict\-ed Alternative Deletion}\xspace}
\newcommand{\problemValueAD}{\textsc{Value-Restricted Alternative Deletion}\xspace}

\newcommand{\problemBconAD}{\textsc{\bcon-Restricted Alternative De\-le\-tion}\xspace}
\newcommand{\problemBconVD}{\textsc{\bcon-Restricted Maverick Deletion}\xspace}

\newcommand{\problemScVD}{\textsc{Sin\-gle-Crossing Maverick Deletion}\xspace}
\newcommand{\problemScAD}{\textsc{Single-Crossing Alternative Deletion}\xspace}

\newcommand{\problemSpVD}{\textsc{Single-Peaked Maverick Deletion}\xspace}

\newcommand{\problemGsVD}{\textsc{Group-Separable Maverick Deletion}\xspace}
\newcommand{\problemGsAD}{\textsc{Group-Separable Alternative Deletion}\xspace}

\newcommand{\problemVC}{\textsc{Vertex Cover}\xspace}
\newcommand{\problemMTwoSAT}{\textsc{Maximum 2-Satisfiability}\xspace}
\newcommand{\problemMTwoSATshort}{\textsc{Max2Sat}}

\newcommand{\acon}{\mbox{\ensuremath{\alpha}}}
\newcommand{\bcon}{\mbox{\ensuremath{\beta}}}
\newcommand{\ccon}{\mbox{\ensuremath{\gamma}}}
\newcommand{\dcon}{\mbox{\ensuremath{\delta}}}

\newcommand{\sub}[1]{_{\scriptscriptstyle{#1}}}

\newcommand{\ppp}{{\cal P}}

\newcommand{\edgevoter}{edge voter}

\newcommand{\worstinconsistentconfig}{worst-diverse configuration}
\newcommand{\Worstinconsistentconfig}{worst-diverse configuration}
\newcommand{\mediuminconsistentconfig}{medium-diverse configuration}
\newcommand{\Mediuminconsistentconfig}{Medium-diverse configuration}
\newcommand{\bestinconsistentconfig}{best-diverse configuration}
\newcommand{\Bestinconsistentconfig}{Best-diverse configuration}
\newcommand{\cyclicconfig}{cyclic configuration}

\newcommand{\binconsistentconfig}{\bcon-configuration}
\newcommand{\dinconsistentconfig}{\dcon-configuration}

\newcommand{\worstr}{worst-restricted}
\newcommand{\mediumr}{medium-res\-trict\-ed}
\newcommand{\bestr}{best-restricted}
\newcommand{\valuer}{value-restricted}
\newcommand{\betar}{\bcon-restricted}

\newcommand{\Worstr}{Worst-restricted}
\newcommand{\Mediumr}{Medium-restricted}
\newcommand{\Bestr}{Best-restricted}

\newcommand{\Valuer}{Value-restricted}

\newcommand\restr[2]{{
  \left.\kern-\nulldelimiterspace 
  #1 
  \vphantom{\big|} 
  \right|_{#2} 
  }}

\newcommand{\decprob}[3]{%
  \begin{center}%
    \begin{minipage}{0.92\linewidth}%
      \textsc{#1}\\
      \textbf{Input:} #2\\
      \textbf{Question:} #3
    \end{minipage}%
  \end{center}%
}

\newcommand{\defvoter}{\text{voter }}
\newcommand{\x}[2]{\ensuremath{\widehat{X}_{#2}^{#1}}}
\newcommand{\sss}{\ensuremath \scriptscriptstyle}

\newcommand{\preferenceorder}{preference order}

\newcommand{\NP}{\textsf{NP}\xspace}
\newcommand{\comP}{\textsf{P}\xspace}

\newcommand{\euclid}{$1$-D Euclidean\xspace}

\newif\ifLineBreak

\newcommand{\configScan}[1]{%
  \ifx\relax#1\empty
  \else
    \ifLineBreak
     \\\relax
    \else
      \LineBreaktrue
    \fi
    #1 \relax
    \LineBreaktrue
  \expandafter\configScan
  \fi
}

\begin{document}


\newcommand{\mytitle}{Are there any nicely structured preference~profiles~nearby?}
\title{\mytitle\thanks{A preliminary short 
version of this work has been presented
at the 23rd International Joint Conference on Artificial Intelligence (IJCAI 2013), Beijing, August, 2013~\cite{BreCheWoe2013}.}}

\author[1]{Robert Bredereck}
\author[1]{Jiehua Chen}
\author[2]{Gerhard J. Woeginger}

\affil[1]{  
  Institut f\"ur Softwaretechnik und Theoretische Informatik,\protect \\
  Technische Universit\"at Berlin,\protect \\
  D-10587 Berlin, Germany\protect \\
  \texttt{\{robert.bredereck, jiehua.chen\}@tu-berlin.de}
}

\affil[2]{
  Department of Mathematics and Computer Science,\protect \\
  TU Eindhoven,\protect \\
  5600 MB Eindhoven, Netherlands\protect \\
  \texttt{gwoegi@win.tue.nl}
}

\pdfinfo{
/Title (\mytitle)
/Author (Robert Bredereck, Jiehua Chen, and Gerhard J. Woeginger) 
}

\date{}

\maketitle

\thispagestyle{plain}

\begin{abstract}
  We investigate the problem of deciding whether a given preference profile is
  close to having a certain nice structure, 
  as for instance single-peaked, single-caved, single-crossing, value-restricted,
  best-restricted, worst-restricted, medium-re\-strict\-ed, or group-separable
  profiles. We measure this distance by the number of voters or
  alternatives that have to be deleted to make the profile a nicely structured
  one. Our results classify the problem variants with respect
  to their computational complexity, and draw a clear line between
  computationally tractable (polynomial-time solvable) and computationally
  intractable (NP-hard) questions.
\end{abstract}

\section{Introduction}\label{sec:intro}
The area of Social Choice (and in particular the subarea of \emph{Computational} Social Choice) is full of 
so-called \emph{negative} results.
On the one hand there are many axiomatic impossibility results, 
and on the other hand there are just as many
computational intractability results.
For instance, the famous impossibility result of 
\citet{Arrow1950} states that there is no perfectly 
fair way (satisfying certain desirable axioms) of aggregating the preferences of a society of voters into a single \preferenceorder.
As another example, \citet{BaToTr1989} establish that it is computationally intractable
(\NP-hard) to determine whether some particular alternative has won an election under a voting scheme designed by Lewis Carroll. 
Most of these negative results hold for general preference profiles where \emph{any} combination of \preferenceorder{s} may occur.

One branch of Social Choice studies \emph{restricted domains} of preference profiles, where only certain nicely
structured combinations of \preferenceorder{s} are admissible.
The standard example for this approach are \emph{single-peaked} preference profiles as introduced by 
\citet{Black1948}: a preference profile is single-peaked if there exists a linear order of the 
alternatives such that every voter's preference along this order is either always strictly increasing, 
always strictly decreasing, or first strictly increasing and then strictly decreasing.
Determining whether a profile is single-peaked is solvable in polynomial time~\cite{BarTri1986,DoiFal1994,EscLanOez2008}.


Single-peakedness implies a number of nice properties, as for instance strat\-e\-gy-proofness of a family of voting rules~\cite{Moulin1980}
and transitivity of the majority relation~\cite{Inada1969}. 
Furthermore, Arrow's impossibility result collapses for single-peaked profiles.
In a similar spirit (but in the algorithmic branch), \citet{Walsh2007}, 
\citet{BraBriHemHem2015}, and 
\citet{FHHR2011} show 
that many electoral bribery, control and manipulation problems that are \NP-hard in the general case become 
tractable under single-peaked profiles.
Besides the single-peaked domain, 
the literature contains many other \emph{restricted domains} of 
nicely structured preference profiles (see \cref{sec:def} for precise mathematical definitions).
\begin{itemize}
\item
\citet{Sen1966} and 
\citet{SePa1969} introduced the domain of \emph{value-restricted} 
preference profiles which satisfy the following:
for every triple of alternatives, 
one alternative is not preferred most by any individual~(\emph{best-restricted} profile), 
or one is not preferred least by any individual~(\emph{worst-restricted} profile),
or one is not considered as the intermediate alternative by any individual~(\emph{medium-restricted} profile).
\item 
\citet{Inada1964,Inada1969} considered the domain of \emph{group-separable} preference profiles 
which satisfy the following:
the alternatives can be split into two groups such that every voter prefers every alternative in the
first group to those in the second group, or prefers every alternative in the second group to those in the first group.
Every group-separable profile is also medium-restricted.
\item
\emph{Single-caved} preference profiles are derived from single-peaked profiles by reversing the preferences
of every voter.
Sometimes single-caved profiles are also called \emph{single-dipped}~\cite{KlPeSt1997}.
\item
\emph{Single-crossing} preference profiles go back to the seminal papers of
 \citet{Mirrlees1971} and \citet{Roberts1977} on income taxation.
A preference profile is single-crossing if there exists a linear order of the voters such that 
each pair of alternatives separates this order into two sub-orders where in each sub-order, 
all voters agree on the relative order of this pair.
Similar to the single-peaked property, 
single-crossing profiles can also be recognized in polynomial time~\cite{BCW13,DoiFal1994,EFS12}.
\end{itemize}

Unfortunately, real-world elections are almost never single-peaked, value-re\-strict\-ed, group-separable,
single-caved or single-crossing.
Usually there are maverick voters whose preferences are determined for instance by race, religion, or gender,
and whose misfit behaviors destroy all nice structures in the preference profile.
In a very recent line of research, 
\citet{FaHeHe2014} searched for a 
cure against such mavericks, and arrived at \emph{nearly} single-peaked preference profiles: a profile is 
nearly single-peaked if it is very close to a single-peaked profile.
Of course there are many mathematical ways of measuring the closeness of profiles.
Natural ways to make a given profile single-peaked are (i) by deletion of voters and (ii) by deletion of alternatives.
This leads to the two central problems of our work for a specific property~$\Pi$ and a given number~$k$: 
\begin{enumerate}
\item The \problemPiVDlong problem asks whether it is possible to delete at most $k$~\emph{voters} to make a given profile satisfy the $\Pi$-property.
\item The \problemPiADlong problem asks whether it is possible to delete at most $k$~\emph{alternatives} to make a given profile satisfy the $\Pi$-property.
\end{enumerate}

In this paper, we have $\Pi \in \{$\worstr, \mediumr, \bestr,
value-restricted, single-peaked, single-caved, single-crossing,
group-separable, \bcon-restricted$\}$.
We provide the formal definitions of these properties in \cref{sec:def}.

\paragraph{Results of this paper.}
We investigate the problem of deciding how close (by deletion of voters and by deletion of alternatives)
a given preference profile is to having a nice structure
(like being single-peaked, or single-crossing, or group-separable).
We focus on the most fundamental definitions of closeness and on the most popular restricted domains.
Our results draw a clear line between computationally easy (polynomial-time solvable) and 
computationally intractable (\NP-hard) questions as they classify all considered problem variants with respect 
to their computational complexity. 
In particular, 
for most of the cases both of our central problems are computationally intractable 
(with the exceptions of maverick deletion when the specific property~$\Pi$ is single-crossing, 
and of alternative deletion when~$\Pi$ is either single-peaked or single-caved).
Our results are summarized in \cref{tab:main_results}.

\begin{table}[t]
  \centering
  \begin{tabular}{@{}l l l  l l l@{}}
    \toprule
    Restriction & \multicolumn{2}{l}{Maverick deletion} & & \multicolumn{2}{l}{Alternative deletion} \\
    \midrule
    &&\\[-0.8em]
    Single-peaked & $\NP$-complete & (*, \mytabcor{cor:sp-sc-gs_vd_npc}) && $\comP$ &(*)\\
    Single-caved & $\NP$-complete & (*, \mytabcor{cor:sp-sc-gs_vd_npc}) && $\comP$ &(*) \\
    Group-separable & $\NP$-complete & (\mytabcor{cor:sp-sc-gs_vd_npc}) && $\NP$-complete &(\mytabcor{cor:gs_ad_npc})\\ 
    Single-crossing & $\comP$ & (\mytabref{thm:sc_vd_p}) && $\NP$-complete &(\mytabref{thm:sc_ad_npc})\\
    &&\\[-.8em]    
    \midrule
    &&\\[-.8em]
    \Valuer & $\NP$-complete &(\mytabref{thm:value-vd_npc}) && $\NP$-complete &(\mytabref{thm:value-ad_npc}) \\
    \Bestr   & $\NP$-complete &(\mytabprop{prop:best-worst-medium-vd_npc})  && $\NP$-complete &(\mytabprop{prop:best-worst-medium-ad_npc})\\
    \Worstr  & $\NP$-complete &(\mytabprop{prop:best-worst-medium-vd_npc}) && $\NP$-complete &(\mytabprop{prop:best-worst-medium-ad_npc})\\ 
    \Mediumr & $\NP$-complete &(\mytabprop{prop:best-worst-medium-vd_npc}) && $\NP$-complete &(\mytabprop{prop:best-worst-medium-ad_npc}) \\
    \bcon-restricted  & $\NP$-complete &(\mytabref{thm:bcon-vd_npc}) && $\NP$-complete &(\mytabref{thm:bcon-ad_npc})\\
    \bottomrule
  \end{tabular}
  \caption{Summary of the results where \comP means polynomial-time solvable.
    Entries marked by ``*'' are due to \citet{ELP13}. 
    The definition of the respective domain restrictions can be found in \cref{sec:def}.}
  \label{tab:main_results}
\end{table}

\paragraph{Related work.}

As to different notions of closeness to restricted domains,
\begin{itemize}
\item 
\citet{ELP13} study various concepts of nearly single-peakedness.
Besides deletion of voters and deletion of alternatives, they also study closeness measures that are based
on swapping alternatives in the preference orders of some voters, or on
introducing additional political axes.
\item \citet{YanGuo14} study \emph{$k$-peaked} domains, 
where every preference order can have at most $k$ peaks (that is, at most $k$ rising streaks that alternate with falling streaks).
\item \citet{CGS12,CGS13} introduce a closeness measure, the \emph{width}, 
for single-peaked, 
single-crossing, 
and group-separable profiles which is based on the notion of a clone set~\cite{Tid1987}.
For instance, the \emph{single-peaked width} of a preference profile is the smallest number~$k$
such that partitioning all alternatives into disjoint \emph{intervals}, each with size at most $k+1$,
and replacing each of these intervals with a single alternative 
results in a single-peaked profile. 
An interval of alternatives is a set of alternatives that 
appear consecutively (in any order) in the preference orders of all voters. 
\end{itemize}

There are several generalizations of the single-peaked property.
For instance, 
\citet{BarGulSta1993} introduce the concept of multi-dimensional single-peaked domains.
The $1$-dimensional special case is equivalent to our single-peaked property.
\citet{SuiFraNieBou2013} study this concept empirically.
They present approximation algorithms~(for several optimization goals) 
of finding multi-dimensional single-peaked profiles 
and show that their two real-world data sets are far from being single-peaked 
but are nearly $2$-dimensional single-peaked. 
While \citet{ELP13} and this paper show that deciding the distance to restricted domains is \NP-hard in most cases,
\citet{ElkLac14} present efficient approximation and fixed-parameter tractable algorithms for deciding the distance to restricted domains such as single-peakedness and single-crossingness.

Finally, we remark that the closeness concept can also be used to characterize voting rules~\cite{Baigent1987,ElkFalSli2012,MesNur2008}. 
The basic idea is to first fix a specific property,
for instance, the transitivity of the pairwise majority relation (also known as Condorcet consistency),
and then to define a closeness measure from a given profile to the ``nearest'' profile with this specific property.
For instance, the Young rule~\citep{You77} 
takes the subprofile that is closest to being Condorcet consistent by deleting the fewest number of voters and selects the corresponding Condorcet winner as a winner; see \citet{ElkFalSli2012} for more information on this.

For restricted and nearly restricted domains, there are various studies on 
single-winner determination~\cite{BraBriHemHem2015}, 
on multi-winner determination~\cite{BetArkJoh2013,SkoYuFalElk2015}, 
on control, manipulation, and bribery~\cite{BraBriHemHem2015,CheFalNieTal2014,FHHR2011,FaHeHe2014}, 
and on possible/necessary winner problems~\cite{Walsh2007}. 
Usually, the expectation is that domain restrictions help in lowering the computational complexity
of many voting problems.
Many publications, however, report that this is not always the case. 
For instance, 
\citet{FaHeHe2014} show that 
the computational complexity of ``controlling approval-based rules'' for nearly single-peaked profiles
is polynomial-time solvable if the distance to single-peaked is a constant, 
and thus, coincides with the one for single-peaked profiles, 
whereas the computational complexity of ``manipulating the veto rule'' for nearly single-peaked profiles
is still \NP-hard and thus, coincides with the one for unrestricted profiles.

\paragraph{Article outline.}
This paper is organized as follows.
\cref{sec:def} summarizes all the basic definitions and notations.
Our results are presented in \cref{sec:value-restricted-results,sec:sp-scaved-gs-results,sec:sc-results}:
\begin{enumerate}
  \item \cref{sec:value-restricted-results} presents results for the \valuer, \bestr, \worstr, and \mediumr{} properties.
  All results are \NP-hardness results and are obtained through reduction from the \NP-complete \problemVC problem (see the beginning of \cref{sec:value-restricted-results} for the definition).
  \item \cref{sec:sp-scaved-gs-results} shows results for single-peakedness, single-cavedness, and group-separability. 
  In addition, this section shows results for the \betar{} property, a necessary condition for group-separability.
  Again, all results are \NP-hardness results and are obtained through reduction from \problemVC.
  \item \cref{sec:sc-results} shows that achieving the single-crossing property by deleting as few alternatives as possible is
  \NP-hard; the reduction is from the \NP-complete \problemMTwoSAT problem~(see the beginning of \cref{sec:sc-results} for the definition), 
  and shows that finding a single-crossing profile with 
  the largest voter set is polynomial-time solvable; this is done by reducing the
  problem to finding a longest path in a directed acyclic graph.
\end{enumerate}
We conclude with some future research directions in the last section.

\begin{figure}[t]
  \centering
  \resizebox{\textwidth}{!}{
    \tikzstyle{class} = [align=center, text centered, rounded
    corners = 2pt, draw, fill=white, rectangle, 
    text depth=.25ex, 
    text height=5ex]
    \tikzstyle{dumstyle} = [text height=5ex]
    \begin{tikzpicture}[>=stealth', shorten >= 1pt, shorten <= 1pt, node distance = 6ex]
      \node[class] (value) at (0,0) {value-\\restricted};
      
      \node[class,  below = 6ex of value] (medium) {medium-\\restricted};
      
      \node[dumstyle, left = 3ex of medium] (dum1) {};
      
      \node[class, below = 4ex of dum1] (baralpha) {$\bar{\alpha}$-\\restricted};
      
      \node[class, left = 3ex of dum1] (best) {best-\\restricted};
      
      \node[dumstyle, right = 3ex of medium] (dum2) {};
      
      \node[class, below = 4ex of dum2] (beta) {$\beta$-\\restricted};
      
      \node[class, right = 3ex of dum2] (worst) {worst-\\restricted};
      
      \node[dumstyle, right = 3ex of worst] (dum3) {};
      
      \node[class, below = 4ex of dum3] (alpha) {$\alpha$-\\restricted};
      
      \node[class, below = 16ex of best] (caved) {single-\\\;\;caved\;\;};
      \node[class, below = 16ex of medium] (gs) {group-\\separable};
      \node[class, below = 16ex of worst] (peaked) {single-\\\;peaked\;};

      \path[-] (best.north) edge (value.south west);
      \path[-] (worst.north) edge (value.south east);
      \path[-] (medium.north) edge (value.south);

      \path[-] (best.south) edge (caved.north);
      \path[-] (baralpha.south) edge (caved.north east);

      \path[-] (medium.south) edge (gs.north);
      \path[-] (beta.south) edge (gs.north east);

      \path[-] (worst.south) edge (peaked.north);
      \path[-] (alpha.south) edge (peaked.north east);

      \node[dumstyle, right = 15ex of dum3] (dum4) {};
      \node[class, below = 4ex of dum4] (gamma) {$\gamma$-\\restricted};
      \node[class, right = 3ex of gamma] (delta) {$\delta$-\\restricted};
      \node[class, below = 16ex of dum4] (crossing) {single-\\crossing};

      \path[-] (crossing.north) edge (gamma.south);
      \path[-] (crossing.north east) edge (delta.south);
    \end{tikzpicture}
  }
  \caption{A Hasse diagram for the relation of the different properties. An edge
  between two properties means that a profile with the property in the lower
  tier implies the property in the upper tier. For instance, there is an
  edge between ``value-restricted'' and ``best-restricted'', because a
  best-restricted profile is also value-restricted.}
  \label{fig:property-relation}
\end{figure}
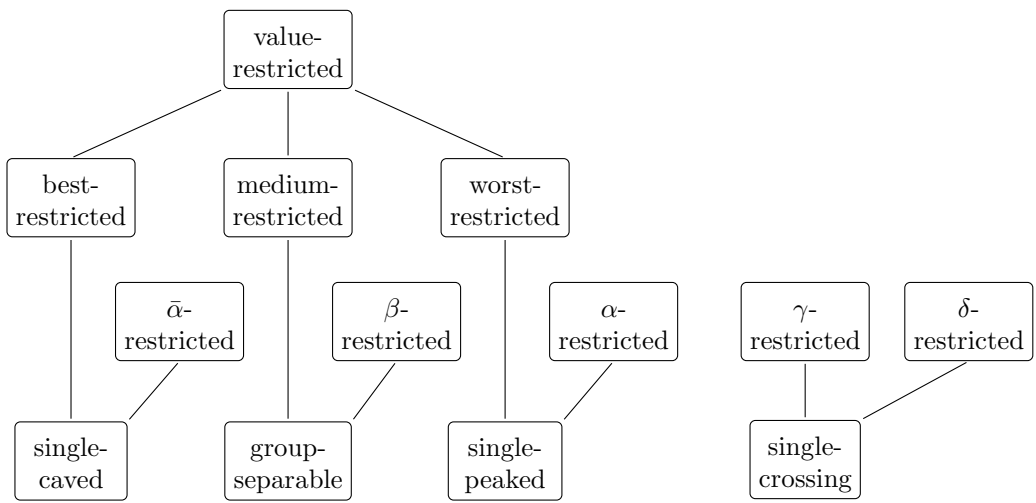

\section{Preliminaries and basic notations}\label{sec:def}

Let $a_{1},\ldots,a_{m}$ be $m$~alternatives and let $v_1,\ldots,v_n$ be $n$~voters. 
A \emph{preference profile} specifies the \emph{\preferenceorder{s}} of the $n$~voters, 
where voter~$v_i$ ranks the alternatives according
to a linear order~$\pref_i$ over the $m$ alternatives. 
For alternatives~$a$ and $b$, the relation
$a\pref_i b$ means that voter~$v_i$ strictly prefers $a$ to~$b$. 
We omit the subscript~$i$ 
if it is clear from the context whose \preferenceorder{} 
we are referring to. 

Given two disjoint sets~$A$ and~$B$ of alternatives, we write $A \pref_i B$
to express that voter~$v_i$ prefers set~$A$ to set~$B$, that is,
for each alternative $a\in A$ and each alternative~$b\in B$ it holds
that~$a\pref_i b$.
We write $A\pref_i b$ as shorthand for $A\pref_i \{b\}$ and $b\pref_i A$ for $\{b\}\pref_i A$.
We sometimes fix a canonical order of the alternatives in~$A$ and denote
this order by~$\seq{A}$.
The expression~$\seq{A_1} \pref \seq{A_2} $ denotes the \preferenceorder{} that is
consistent with $\seq{A_1}$, $\seq{A_2}$, and $A_1 \pref A_2$.

Next, we review some preference profiles with special properties studied in the
literature~\cite{BH11,BCW13,Inada1964,Inada1969,SePa1969,Sen1966}. 
We call such profiles \emph{configurations}
and we use them to characterize some properties of the preference profiles.
We illustrate the relation between the respective properties in \cref{fig:property-relation}.

\subsection{Value-restricted profiles}
The first three configurations~\cite{BH11} describe profiles with three alternatives where
each alternative is in the best, medium, or worst position
in some voter's \preferenceorder.

\begin{definition}[\Bestinconsistentconfig]
  A profile with three voters~$v_1, v_2, v_3$, and three distinct
  alternatives~$a$, $b$, $c$ is a \emph{\bestinconsistentconfig} if it
  satisfies the following:
  \begin{quote}
    $\defvoter v_1 \colon~ a \pref \{b, c\}$,\\
    $\defvoter v_2 \colon~ b \pref \{a, c\}$,\\
    $\defvoter v_3 \colon~ c \pref \{a, b\}$. 
  \end{quote}
\end{definition}

\begin{definition}[\Mediuminconsistentconfig]
  A profile with three voters~$v_1, v_2,$ $v_3$, and three distinct
  alternatives~$a, b, c$ is a \emph{\mediuminconsistentconfig} if it
  satisfies the following:
  \begin{quote}
    $\defvoter v_1\colon~b \pref a \pref c \;\text{ or }\; c \pref a \pref b$,\\
    $\defvoter v_2\colon~a \pref b \pref c \;\text{ or }\; c \pref b \pref a$,\\
    $\defvoter v_3\colon~a \pref c \pref b \;\text{ or }\; b \pref c \pref a$.
  \end{quote}
\end{definition}

\begin{definition}[\Worstinconsistentconfig]
  A profile with three voters~$v_1, v_2, v_3$ and three distinct
  alternatives~$a$,$b$,$c$ is a \emph{\worstinconsistentconfig} if it
  satisfies the following:
  \begin{quote}
    $\defvoter v_1\colon~ \{b, c\} \pref a$,\\
    $\defvoter v_2\colon~ \{a, c\} \pref b$,~\\
    $\defvoter v_3\colon~ \{a, b\} \pref c$.
  \end{quote}
\end{definition}

We use these three configurations to characterize several restricted domains:
A profile is \emph{\bestr} 
(resp.\ \emph{\mediumr}, 
\emph{\worstr}) 
\emph{with respect to a triple of alternatives} 
if it contains no three voters that form a \bestinconsistentconfig{} 
(resp.\ a \mediuminconsistentconfig, a \worstinconsistentconfig) with respect to this triple. 
A profile is \emph{\bestr} (resp.\ \emph{\mediumr}, \emph{\worstr})
if it is \bestr{} (resp.\ \mediumr, \worstr) with respect to every possible triple of alternatives. 


A profile is \emph{value-restricted}~\cite{Sen1966} 
if for every triple~$T$ of alternatives, 
it is \bestr, \mediumr, or \worstr{} with respect to $T$.
In other words, 
a profile is \emph{not} value-restricted 
if and only if it contains a triple of alternatives and three voters~$v_1, v_2, v_3$ 
that form a \bestinconsistentconfig, a \mediuminconsistentconfig, 
and a \worstinconsistentconfig with respect to this triple.

\begin{definition}[Cyclic configuration]
  A profile with three voters~$v_1, v_2, v_3$ and three distinct
  alternatives~$a$,$b$,$c$ is a \emph{cyclic configuration} if it
  satisfies the following:
\begin{quote}
  $\defvoter v_1\colon~ a \pref b \pref c$,\\
  $\defvoter v_2\colon~ b \pref c \pref a$,\\
  $\defvoter v_3\colon~ c \pref a \pref b$.
\end{quote}
\end{definition}

The set of the preference orders of all voters in a \valuer{} profile
is also known as \emph{acyclic domain of linear orders}.
Many research groups~\citep{AbeJoh1984,GalRei2008,Monjardet2009,PupSli2015} 
investigate maximal acyclic domains for a given number~$m$ of alternatives,
where an acyclic domain is \emph{maximal} if adding any new linear order destroys the \valuer{} property.

\subsection{Single-peaked profiles and single-caved profiles}
Given a set~$A$ of alternatives and a linear order~$L$ over~$A$, 
we say that a voter~$v$ is \emph{single-peaked} with respect to $L$
if his preference along $L$ is always strictly increasing, 
always strictly decreasing, or first strictly increasing and then strictly decreasing.
Formally, a voter~$v$ is single-peaked with respect to $L$ if
for each three distinct alternatives~$a,b,c\in A$, it holds that
\begin{align*}
 (a \pref_L b \pref_L c \text{ or } c \pref_L b \pref_L a) \text{ implies that if } a \pref_v b, \text{ then } b \pref_v c.
\end{align*}
A profile with the alternative set~$A$ 
is \emph{single-peaked} if there is a linear order over~$A$ such that
every voter is single-peaked with respect to this order.
Single-peaked profiles are necessarily \worstr.
To see this,
we observe that in a profile with at least three alternatives, 
the alternative that is ranked last by at least one voter must not be placed between the other two along any single-peaked order.
But then, none of the alternatives~$a$, $b$, and $c$ from a \worstinconsistentconfig{} can be placed between the other two in any single-peaked order.
Thus, a profile with \worstinconsistentconfig{s} cannot be single-peaked.

To fully characterize the single-peaked domain, we additionally need the following configuration.

\begin{definition}[\acon-configuration]\label{ex:alpha}
A profile with two voters~$v_1$ and $v_2$, and four distinct alternatives~$a, b, c, d$ is an \acon-configuration if it satisfies the following:
\begin{quote}
$\defvoter v_1\colon~ \{a,d\}\pref b \pref c$,\\
$\defvoter v_2\colon~ \{c,d\}\pref b \pref a$.
\end{quote}
\end{definition}

The \acon-configuration describes a situation 
where two voters have opposite opinions on the order of three alternatives~$a, b$ and $c$ 
but agree that a fourth alternative~$d$ is ``better'' than the one ranked in the middle.
A profile with this configuration is not single-peaked 
as we must put alternatives~$b$ and $d$ between alternatives~$a$ and $c$, 
but then voter~$v_1$ prevents us from putting $b$ next to $a$ 
and voter~$v_2$ prevents us from putting $b$ next to $c$.

A profile is \emph{single-peaked} if and only if it contains neither
\worstinconsistentconfig{}s nor \acon-configurations~\cite{BH11}.
Since reversing the \preferenceorder{s} of a single-peaked profile results in
a single-caved one,
an analogous characterization of single-caved profiles follows.
A profile is \emph{single-caved} if and only if it contains neither
\bestinconsistentconfig{}s nor $\bar{\acon}$-configurations
where a $\bar{\acon}$-configuration is a \acon-configuration with
both ``\preferenceorder{s}'' being inverted:

\begin{definition}[$\bar{\acon}$-configuration]\label{ex:alpha-bar}
A profile with two voters~$v_1$ and $v_2$, and four distinct alternatives
$a,b,c,d$ is an $\bar{\acon}$-configuration if it satisfies the following:
\begin{quote}
$\defvoter v_1 \colon~ a\pref b \pref \{c,d\}$,\\
$\defvoter v_2 \colon~ c\pref b \pref \{a,d\}$.
\end{quote}
\end{definition}

\subsection{Group-separable profiles}
Given a profile with $A$ being the set of alternatives,
the \emph{group-separable} property requires that 
every size-at-least-three subset~$A'\subseteq A$ can be partitioned into two disjoint non-empty subsets~$A'_1$ and $A'_2$
such that for each voter $v_i$, either~$A'_1 \pref_i A'_2$ or $A'_2 \pref_i A'_1$ holds.
One can verify that group-separable profiles are necessarily medium-restricted.
\citet{BH11} characterized the group-separable property using the following configuration.

\begin{definition}[\binconsistentconfig]\label{ex:beta}\ \\
  A profile with two voters~$v_1$ and $v_2$ and four distinct
  alternatives~$a,b,c,d$ is a \bcon-configuration if it satisfies the
  following:
\begin{quote}
$\defvoter v_1\colon~ a\pref b \pref c \pref d$, \\
$\defvoter v_2\colon~ b\pref d \pref a \pref c$.
\end{quote}
\end{definition}

The \bcon-configuration describes a situation where
the most preferred alternative 
and the least preferred alternative of voter~$v_1$ are $a$ and $d$ which
are different from the ones of voter~$v_2$: $b$ and $c$.
Both voters agree that $b$ is better than $c$, but disagree on
whether $d$ is better than $a$. 
This profile is not group-separable:
We can not partition $\{a,b,c,d\}$ into
one singleton
and one three-alternatives set as each alternative is ranked in the middle once, 
but neither can we partition them into two sets each of size two
since voter $v_1$ prevents us from putting alternatives~$a$ and~$c$ 
or alternatives~$a$ and~$d$ together 
and voter $v_2$ prevents us from putting alternatives~$a$ and $b$ together. 
Profiles without \bcon-configurations are called \emph{\bcon-restricted}~\cite{BH11}.

A profile is \emph{group-separable} if and only if it contains neither
\mediuminconsistentconfig{s} nor \bcon-configurations~\cite{BH11}.

\subsection{Single-crossing profiles} 
The single-crossing property describes the existence of a ``natural'' linear order of the voters. 
A preference profile is \emph{single-crossing} 
if there exists a \emph{single-crossing order} of the voters, that is, 
a linear order~$L$ of the voters, such that 
each pair of alternatives separates $L$ into two sub-orders where in each sub-order, 
all voters agree on the relative order of this pair.
Formally, this means that for each pair of alternatives~$a$ and $b$ 
such that the first voter along the order~$L$ 
prefers $a$ to $b$ and for each two voters~$v, v'$ with $v\pref_L v'$,
\begin{align*}
  b \pref_{v} a \text{ implies } b \pref_{v'} a \text{.}
\end{align*}

To characterize single-crossing profiles, 
we need the following two configurations.

\begin{definition}[\ccon-configuration]\label{ex:gamma}\ \\
A profile with three voters $v_1, v_2, v_3$,
and six (not necessarily distinct) alternatives
$a,b,c,d,e,f$ is a \ccon-configuration, if it satisfies the following:
\begin{quote}
$\defvoter~v_1\colon~b\pref a$ ~and~ $c\pref d$ ~and~ $e\pref f$,\\
$\defvoter~v_2\colon~a\pref b$ ~and~ $d\pref c$ ~and~ $e\pref f$,\\
$\defvoter~v_3\colon~a\pref b$ ~and~ $c\pref d$ ~and~ $f\pref e$.
\end{quote}
\end{definition}

The \ccon-configuration describes a situation 
where each voter disagrees with the other two voters on the order of exactly two distinct alternatives. 
The profile is not single-crossing 
as none of the three voters can be put between the other two: The pair~$\{a,b\}$
prevents us from putting $v_1$ in the middle, 
the pair~$\{c,d\}$ forbids voter~$v_2$ in the middle, 
and the pair~$\{e,f\}$ forbids $v_3$ in the middle.

\begin{definition}[\dcon-configuration]\label{ex:delta} \ \\
A profile with four voters $v_1, v_2, v_3, v_4$,
and four (not necessarily distinct)
alternatives $a,b,c,d$ is a \dcon-configuration, if it satisfies the following:
\begin{quote}
$\defvoter v_1\colon~ a\pref b$ ~and~ $c\pref d$,\\
$\defvoter v_2\colon~ a\pref b$ ~and~ $d\pref c$,\\
$\defvoter v_3\colon~ b\pref a$ ~and~ $c\pref d$,\\
$\defvoter v_4\colon~ b\pref a$ ~and~ $d\pref c$.
\end{quote}
\end{definition}
The \dcon-configuration shows a different kind of voter behavior: 
Two voters disagree with the other two voters on the order of two alternatives,
but also disagree between each other on the order of two further alternatives. 
This profile is not single-crossing
as the pair~$\{a,b\}$ forces us to place $v_1$ and $v_2$ next to each other, and to
put $v_3$ and $v_4$ next to each other; the pair~$\{c,d\}$ forces us to place
$v_1$ and~$v_3$ next to each other, and to put $v_2$ and $v_4$ next to each
other.  This means that no voter can be placed in the first position.

A profile is \emph{single-crossing} if and only if it contains neither
\ccon-configurations nor \dcon-configura\-tions~\cite{BCW13}.

\subsection{Two central problems}\label{sub:central problems}
As already discussed before, two natural ways of measuring the closeness of
profiles to some restricted domains are to count the number of voters
resp. alternatives which have to be deleted to make a profile
single-crossing. Hence, for $\Pi \!\in\!\{$\worstr, \mediumr, \bestr,
value-restricted, single-peaked, single-caved, single-crossing,
group-separable, \bcon-restricted$\}$, we study the following two decision
problems: \problemPiVD and \problemPiAD.

\decprob{\problemPiVDlong} 
{A profile with $n$~voters and a non-negative integer $k\le n$.}
{Can we delete at most $k$ voters so that the resulting profile satisfies the $\Pi$-property?}

\decprob{\problemPiADlong} 
{A profile with $m$ alternatives and a non-negative integer~$k\le m$.}
{Can we delete at most $k$ alternatives so that the resulting profile satisfies the $\Pi$-property?}


An upper bound for the computational complexity of \problemPiVD and \problemPiAD
is easy to see.
Both problems are contained in~\NP for each property~$\Pi$ we study:
Given a preference profile one can check in
polynomial time whether it has property~$\Pi$
since $\Pi$ is characterized
by a finite set of forbidden finite substructures. 
Thus, in order to show \NP-completeness of \problemPiVD and \problemPiAD, 
we only have to show their \NP-hardness.

\section{Value-restricted properties}
\label{sec:value-restricted-results}
In this section, we show \NP-hardness for the \valuer,  \bestr, \worstr, and \mediumr{} domains, respectively.
Notably, we show all these results by reducing from the \NP-complete \textsc{Vertex Cover} problem~\cite{GJ79}.

\decprob{Vertex Cover}
{An undirected graph~$G = (U, E)$ and a non-negative integer~$k \le |U|$.}
{Is there a \emph{vertex cover}~$U'\subseteq U$ of at most~$k$ vertices, that is, 
  $|U'| \le k$ and $\forall e \in E \colon e \cap U' \neq \emptyset$?}
In every reduction from \textsc{Vertex Cover} we describe, the vertex cover size~$k$ coincides with the
maximum number~$k$ of voters (resp.\ alternatives) to delete.
Hence, we use the same variable name.

We first deal with the case of maverick voter deletion~(\cref{subsec:reduction from vc-maverick deletion})
and then, 
with the case of deleting alternatives (\cref{subsec:reduction from vc-alternative deletion}).
In both cases, 
the general idea is to transform every edge of a given graph into an appropriate forbidden configuration.

\subsection{Maverick Voter Deletion}\label{subsec:reduction from vc-maverick deletion}

\begin{theorem}\label{thm:value-vd_npc}
  \problemValueVD is \NP-complete.
\end{theorem}

\begin{proof}
  We provide a polynomial-time reduction from \problemVC to show \NP-hardness.
  We will present an example (see \cref{fig:vc-graph-reduced-instance}) for the reduction right after this proof.

  Let $(G=(U, E), k)$ denote a \problemVC instance with vertex set~$U=\{u_{1}, \ldots, u_r\}$ and edge set~$E=\{e_{1}, \ldots, e_s\}$;
  without loss of generality we assume that the input graph~$G$ is connected and
  that graph~$G$ has at least four vertices, that is, $r\ge4$, and that $k \le r-3$.
  
  The set of alternatives consists of three \emph{edge alternative}s~$a_j, b_j$,
  and $c_j$ for each edge~$e_j \in E$.
  For each vertex in $U$, we construct one voter.
  That is, we define $A \coloneqq \{a_j,b_j,c_j \mid e_j \in E\}$ and $V \coloneqq \{v_i \mid u_i \in U\}$.
  In total, the number~$m$ of alternatives is~$3s$ and the number~$n$
  of voters is~$r$.

  Now we specify the \preferenceorder{} of each voter.
  Every voter prefers~$\{a_{j}, b_{j},$ $c_{j}\}$ to $\{a_{j'}$, $b_{j'}$, $c_{j'}\}$ whenever $j < j'$.
  For each edge~$e_j$ with two incident vertices~$u_i$ and $u_{i'}$, $i<i'$,
  and for each non-incident vertex~$u_{i''} \notin e_j$,
  the following holds:
  \begin{alignat*}{2}    
    &\defvoter v_i   &\colon &\;\; c_j \pref a_j \pref b_j, \\
    &\defvoter v_{i'} &\colon &\;\;  b_j \pref c_j \pref a_j,\\
    &\defvoter v_{i''} &\colon &\;\; a_j \pref b_j \pref c_j.
  \end{alignat*}
  In this way, the two vertex voters that correspond to the vertices in $e_j$ and any voter~$v_{z}$ not in $e_j$ form a \cyclicconfig{} with regard to the three edge alternatives~$a_j, b_j, \text{ and } c_j$.
  By the definition of \cyclicconfig{s}, this configuration is also best-diverse, medium-diverse, and worst-diverse.

  The maximum number of voters to delete equals the maximum vertex cover size~$k$.
  This completes the construction which can be done in polynomial time.

  It remains to show its correctness.
  In particular, we show that~$(G=(U, E), k)$ has a vertex cover
  of size at most~$k$ if and only if the constructed profile can be made
  \valuer{} by deleting at most~$k$ voters.

  For the ``only if'' part, suppose that $U'\subseteq U$ with $|U'|\le k$ is a vertex cover.
  We show that, after deleting the voters corresponding to the vertices in $U'$ 
  the resulting profile is \valuer.
  Suppose for the sake of contradiction that the resulting profile is 
  \emph{not} \valuer.
  That is, it still contains 
  a \cyclicconfig{}~$\sigma$.
  By the definition of \cyclicconfig{s}, it must hold that 
  for each pair of alternatives~$x$ and $y$ in $\sigma$,
  there are two voters, one preferring~$x$ to $y$ and the other preferring $y$ to $x$.
  Together with the fact that all voters agree on the relative order of 
  two edge alternatives that correspond to different edges, 
  this implies that 
  the three alternatives~$a_j$, $b_j$, and $c_j$ in $\sigma$
  correspond to the same edge~$e_j$.
  Furthermore, $\sigma$~involves two voters corresponding to the incident vertices from $e_j$,
  and one other voter,
  because all voters corresponding to vertices 
  not in $e_j$ have the same ranking~$a_j \pref b_j \pref c_j$.
  Then, edge~$e_j$ is not covered by any vertex in~$U'$---a contradiction.

  For the ``if'' part, suppose that the profile becomes \valuer{} after the removal of a subset~$V'\subseteq V$ of voters with $|V'|\le k$.
  That is, no three remaining voters form a \cyclicconfig{}.
  We show by contradiction that~$V'$ corresponds to a vertex cover of graph~$G$.
  Assume towards a contradiction that an edge~$e_j$ is not covered
  by the vertices corresponding to the voters in~$V'$.
  Then, the two voters corresponding to the vertices that are incident with
  edge~$e_j$ together with a third voter form a \cyclicconfig{}
  with regard to the three alternatives~$a_j$, $b_j$ and $c_j$---a contradiction.
  Thus, $V'$ corresponds to a vertex cover of graph~$G$ and its size is at most $k$.
\end{proof}

\tikzstyle{sol}=[circle, minimum size=18pt, inner sep = 1.5pt,draw, fill=black!10, font=\small]
\tikzstyle{vertex} = [circle, minimum size=18pt, inner sep = 1.5pt, draw, font=\small]
\tikzstyle{line}=[draw=black,-]
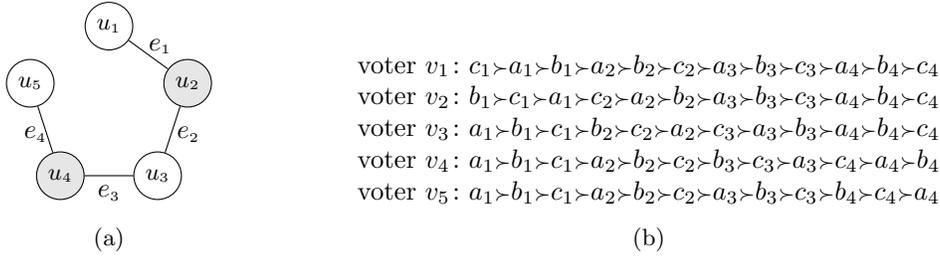
\begin{figure}
  \centering
  \begin{subfigure}[b]{.25\textwidth}
    \centering
    \begin{tikzpicture}
      %
      %
      %
      %
      \draw \foreach \x/\y/\t in {90/1/vertex,18/2/sol,306/3/vertex,234/4/sol,162/5/vertex} {
        (\x:1.1) node[\t] (v\y) {$u_{\sss \y}$} 
      };
      
      \draw (v1) edge[line] node[above,pos=.8] (e1) {$e_{\sss 1}$} (v2);
      \draw (v2) edge[line] node[pos=.6,xshift=1.5ex] (e2) {$e_{\sss 2}$} (v3);
      \draw (v3) edge[line] node[midway,yshift=-1.5ex] (e3) {$e_{\sss 3}$} (v4);
      \draw (v4) edge[line] node[pos=.4,xshift=-1ex] (e4) {$e_{\sss 4}$} (v5);
      
    \end{tikzpicture}
    \caption{}\label{fig:vc-graph-reduced-instance-graph}
  \end{subfigure}
  \qquad
  \begin{subfigure}[b]{.65\textwidth}
    \centering
    
    \begin{tabular}{@{}l@{}}
      $\defvoter v_1\colon c_1 \spref a_1 \spref b_1 \spref a_2 \spref b_2 \spref c_2 \spref a_3 \spref b_3 \spref c_3 \spref a_4 \spref b_4 \spref c_4$\\
      $\defvoter v_2\colon b_1 \spref c_1 \spref a_1 \spref c_2 \spref a_2 \spref b_2 \spref a_3 \spref b_3 \spref c_3 \spref a_4 \spref b_4 \spref c_4$\\
      $\defvoter v_3\colon a_1 \spref b_1 \spref c_1 \spref b_2 \spref c_2 \spref a_2 \spref c_3 \spref a_3 \spref b_3 \spref a_4 \spref b_4 \spref c_4$\\
      $\defvoter v_4\colon a_1 \spref b_1 \spref c_1 \spref a_2 \spref b_2 \spref c_2 \spref b_3 \spref c_3 \spref a_3 \spref c_4 \spref a_4 \spref b_4$\\
      $\defvoter v_5\colon a_1 \spref b_1 \spref c_1 \spref a_2 \spref b_2 \spref c_2 \spref a_3 \spref b_3 \spref c_3 \spref b_4 \spref c_4 \spref a_4$
    \end{tabular}
    \caption{}\label{fig:vc-graph-reduced-instance-profile}
  \end{subfigure}
    \caption{(a)
      An undirected graph with $5$ vertices and $4$ edges. The graph has a vertex cover of size~$2$ (filled in gray).
      (b) A reduced instance~$((A, V), k=2)$ of \problemValueVD, 
      where $A = \{a_i, b_i, c_i \mid 1\le i \le 4\}$ and $V=\{v_1, v_2, v_3, v_4, v_5\}$.
      Deleting $v_2$ and $v_4$ results in a \valuer{} profile.
      In fact, the resulting profile is also \bestr, single-peaked (and hence \worstr), and group-separable (and hence \mediumr).}
    \label{fig:vc-graph-reduced-instance}
\end{figure}

We illustrate our reduction through an example. 
\cref{fig:vc-graph-reduced-instance}(\subref{fig:vc-graph-reduced-instance-graph}) 
depicts an undirected graph with $5$ vertices and $4$ edges.
Vertices $u_2$ and $u_4$ form a vertex cover of size two.
\cref{fig:vc-graph-reduced-instance}(\subref{fig:vc-graph-reduced-instance-profile})
shows the reduced instance with~$5$ voters and $3\cdot 4 = 12$ alternatives.
Deleting voters~$v_2$ and~$v_4$ results in a \valuer{} profile which is also \bestr, single-peaked (and hence \worstr), 
and group-separable (and hence \mediumr).

\medskip
Taking a closer look at the reduction shown in the proof of \cref{thm:value-vd_npc},
the constructed profile contains \cyclicconfig{s} which are simultaneously best-diverse, worst-diverse, and medium-diverse.
It turns out that we can use the same construction to show the following three \NP-hardness results with regard to the \bestr, \worstr, and \mediumr{} properties.

\begin{prop}\label{prop:best-worst-medium-vd_npc}
  \problemPiVDlong is \NP-complete for every property~$\Pi \in \{$\bestr,
    \worstr, \mediumr$\}$.
\end{prop}

\begin{proof}
  Let $(G=(U, E),k)$ be a \problemVC instance and
  let $A$ and $V$ be the set of alternatives and the set of voters that are constructed in the same way as in the proof of \cref{thm:value-vd_npc}.
  Let $k$ be the number of voters to be deleted.
  As we already observed in that proof,
  for each edge~$e_j \in E$,
  the two vertex voters that correspond to the vertices in $e_j$ and any other
  voter~$v_z$ form a cyclic configuration, that is, a best-diverse, worst-diverse, and \mediuminconsistentconfig.
  It remains to show that 
  $(G=(U, E), k)$ has a vertex cover
  of size at most~$k$ if and only if the constructed profile can be made
  \bestr{} (or \worstr{} or \mediumr) by deleting at most~$k$ voters.
  
  For the ``if'' part, suppose that the profile becomes \bestr{} (or \worstr{} or \mediumr{})
  by deleting a subset~$V'\subseteq V$ of voters with $|V'|\le k$.
  Then, the resulting profile is also \valuer.
  Thus, we can use the ``if'' part in the proof of \cref{thm:value-vd_npc} 
  and obtain that the vertices corresponding to $V'$ form a vertex cover of size at most~$k$.

  For the ``only if'' part, suppose that $U'\subseteq U$ with $|U'|\le k$ is a vertex cover.
  As in the ``only if'' part proof of \cref{thm:value-vd_npc},
  we can show that deleting the voters corresponding to $U'$ results in a \bestr{} and a \worstr{} profile.
  
  As for the \mediumr{} property,
  suppose towards a contradiction that after deleting the voters corresponding to the vertices in $U'$ 
  there is still a \mediuminconsistentconfig{}~$\sigma$.
  By the definition of \mediuminconsistentconfig{s}, 
  we know that in $\sigma$, each alternative is ranked between the other two by one voter.
  This implies that $\sigma$ involves three alternatives that correspond to the same edge~$e_j$ and
  involves two voters that correspond to $e_j$'s incident vertices.
  Thus, $e_j$ is an uncovered edge---a contradiction.
\end{proof}

\subsection{Alternative Deletion}\label{subsec:reduction from vc-alternative deletion}

Next, we consider the case of deleting alternatives.
Just as for the voter deletion case, we first show \NP-hardness of deciding the distance
to \valuer{} profiles.
Then, we show how to adapt the reduction to also work for
deciding the distance to \bestr{}, \worstr{}, and \mediumr{} profiles, respectively.

\begin{theorem}\label{thm:value-ad_npc}
  \problemValueAD is \NP-complete.
\end{theorem}

\begin{proof}
  We reduce from \problemVC to \problemValueAD.
  Let $(G=(U, E), k)$ denote a \problemVC instance with vertex set~$U=\{u_{1}, \ldots, u_r\}$ 
  and edge set~$E=\{e_{1}, \ldots, e_s\}$.
  The set of alternatives consists of one \emph{vertex alternative}~$a_j$ for each vertex~$u_j$ in $U$
  and of $k+1$ additional \emph{dummy alternatives}. 
  Let $A$ denote the set of all vertex alternatives and
  let $D$ denote the set of all dummy alternatives.
  We arbitrarily fix a canonical order~$\seq{D}$ of~$D$ and
  we set $\seq{A} \coloneqq a_{1} \pref a_{2} \pref \ldots \pref a_{r}$.
  The number~$m$ of constructed alternatives is $r+k+1$.

  We introduce a voter~$v_{0}$ with the canonical \preferenceorder~$\seq{D}\pref \seq{A}$.
  For each edge~$e_i = \{u_{j},  u_{j'}\}$ with $j < j'$, we introduce
  two \emph{\edgevoter{s}} $v_{2i-1}$ and $v_{2i}$ with \preferenceorder{s}
  \begin{alignat*}{3}
    &\defvoter v_{2i-1}    &\colon &\;\;  a_j \pref a_{j'} \pref \seq{D} \pref \seq{A \setminus \{a_j, a_{j'}\}}, \\
    &\defvoter v_{2i}      &\colon &\;\;  a_{j'} \pref \seq{D} \pref  \seq{A\setminus \{a_{j'}\}}\text{.}
  \end{alignat*}
  Together with voter~$v_0$, the two voters~$v_{2i-1}$ and~$v_{2i}$ form a
  \cyclicconfig{} with respect to the two vertex alternatives~$a_{j}$, $a_{j'}$ and an arbitrary dummy alternative from $D$.
  Let $V$ denote the set of all voters.
  In total, the number~$n$ of constructed voters is $2s+1$.
  The maximum number of alternatives to delete equals the maximum vertex cover size~$k$.
  This completes the construction.

  Our reduction runs in polynomial time.  
  It remains to show that graph~$G$ has a vertex cover of size at most~$k$ 
  if and only if the constructed profile can be made \valuer{} by deleting at most $k$~alternatives. 

  For the ``only if'' part, suppose that $U'\subseteq U$ with $|U'|\le k$ is
  a vertex cover.
  We show that after deleting the vertex alternatives corresponding to~$U'$,
  the resulting profile is \valuer{}.
  Suppose for the sake of contradiction that the resulting profile is not \valuer, 
  that is, it contains a \cyclicconfig~$\sigma$.
  By definition, it must hold that 
  for each pair of alternatives~$x$ and $y$ in $\sigma$,
  there are two voters, one preferring $x$ to $y$ and the other preferring $y$ to $x$.
  Together with the fact that each voter agrees on the relative order of two distinct dummy alternatives, 
  this implies that $\sigma$ contains at most one dummy alternative.
  But if $\sigma$ contains one dummy alternative~$d\in D$, 
  then there is a voter with preferences~$a_j \pref a_{j'} \pref d$ where $a_j, a_{j'} \in A$, 
  which means that edge~$\{u_j, u_{j'}\}$ is not covered by~$U'$.
  Hence, $\sigma$ contains no dummy alternative.
  This means that~$\sigma$ contains three vertex alternatives $a_{j}$, $a_{j'}$, and $a_{j''}$
  with $j< j' < j''$ and by the definition of \cyclicconfig{s},
  $\sigma$ involves three voters with preferences $\{a_{j}, a_{j'}\}  \pref  a_{j''}$,
  $\{a_{j}, a_{j''}\}  \pref  a_{j'}$, and $\{a_{j'}, a_{j''}\}  \pref  a_{j}$, respectively.
  However, the last preference implies that $\{u_{j'}, u_{j''}\}$ is an
  edge which is not covered by $U'$---a contradiction.

  For the ``if'' part, suppose that the constructed profile is a yes-instance of \problemValueAD.
  Let $A'\subseteq A\cup D$ be the set of deleted vertex alternatives with $|A'|\le k$.
  We show that the vertex set~$U'$ corresponding to $A'$ form a vertex cover of graph~$G$ and has size at most $k$.
  Clearly, $|U'|\le k$.
  Assume towards a contradiction that $e_i=\{u_{j}, u_{j'}\}, j<j'$, 
  is not covered by $U'$.
  Since $|D| > k$, at least one dummy alternative~$d$ is not deleted.
  Then, $v_0$ and the two edge voters~$v_{2i-1}, v_{2i}$ form a
  \cyclicconfig{} with regard to $a_{j}, a_{j'}, d$---a contradiction.
\end{proof}

Using the same construction as in the last proof,
we can show that achieving best-restriction, worst-restriction, 
or medium-restriction via deleting the fewest number of alternatives is intractable.

\begin{prop}\label{prop:best-worst-medium-ad_npc}
  \problemPiAD is \NP-complete for every property~$\Pi\in\{$\bestr{}, \worstr{}, \mediumr{}$\}$.
\end{prop}

\begin{proof}
  Let $((U, E), k)$ denote a \problemVC instance with vertex set~$U=\{u_{1}, \ldots, u_r\}$ 
  and edge set~$E=\{e_{1}, \ldots, e_s\}$.
  Let $A$, $D$, and $V$ be the sets constructed in the same way as in the proof of \cref{thm:value-ad_npc}.

  It remains to show that $((U, E), k)$ has a vertex cover of size at most~$k$ 
  if and only if the constructed profile can be made \bestr{} (or \worstr{} or \mediumr{}) by deleting at most $k$~alternatives. 
 
  For the ``if'' part, suppose that the profile becomes \bestr{} (or \worstr{} or \mediumr{}) after deleting a set~$A'\subseteq A\cup D$  of at most~$k$~alternatives.
  Then, the resulting profile is also \valuer.
  Thus, we can use the ``if'' part in the proof of \cref{thm:value-ad_npc} and obtain that the vertices corresponding to $A'$ form a vertex set of size at most~$k$.
  
  For the ``only if'' part, suppose that $U'\subseteq U$ with $|U'|\le k$ is a vertex cover.
  Just as in the ``only if'' part proof of \cref{thm:value-ad_npc},
  we can show that deleting the alternatives corresponding to $U'$ results in a \bestr{} and a \worstr{} profile.
  
  Now, we consider the \mediumr{} property.
  Suppose for the sake of contradiction that the resulting profile is not \mediumr,
  that is, it contains a \mediuminconsistentconfig~$\sigma$.
  Since all voters rank $\seq{D}$ 
  and since no voter rank $d \pref a_j \pref d'$ with $d,d' \in D$ and $a_j \in A$, 
  configuration~$\sigma$ contains at most one dummy alternative.
  Now, if $\sigma$ involves one dummy alternative~$d\in D$ and two vertex alternatives~$a_{j}, a_{j'} \in A$ with $j < j'$, 
  then the voter ranking $a_{j'}$ between $a_{j}$ and $d$ must rank $a_j \pref a_{j'} \pref d$. 
  But this means that edge~$\{u_j,u_{j'}\}$ is uncovered---a contradiction. 
  Hence, $\sigma$ contains no dummy alternative. 
  This means that $\sigma$ involves three vertex alternatives~$a_{j}, a_{j'}, a_{j''}$ with $j <  j' < j''$.
  By the definition of \mediuminconsistentconfig{s}, $\sigma$ must contain a voter that ranks $a_{j''}$ between $a_{j}$ and $a_{j'}$,
  that is, 
  a voter ranks either $a_{j}\pref a_{j''} \pref a_{j'}$ or $a_{j'}\pref a_{j''} \pref a_{j}$.
  This, however, implies that either edge~$\{u_j, u_{j''}\}$ or edge~$\{u_{j'}, u_{j''}\}$ is uncovered--a contradiction.
\end{proof}

\section{Single-peaked, single-caved, and group-separable properties}
\label{sec:sp-scaved-gs-results}

Since single-peaked, group-separable, and single-caved profiles are necessarily
\worstr{}, \mediumr, and \bestr, respectively, it seems reasonable
to expect that the intractability result (\cref{prop:best-worst-medium-vd_npc}) transfers.
Indeed, we can show that this immediately follows from the proofs of \cref{prop:best-worst-medium-vd_npc} (and hence of \cref{thm:value-vd_npc}) because
the profile constructed in the \NP-hardness reduction contains neither \acon-configurations,
nor \bcon-configurations, nor $\bar{\acon}$-configurations.
Note that \NP-hardness of \problemSpVD is already known by a different proof of \citet{ELP13}.
However, their proof does not work for \problemPiVD with $\Pi \in \{$\bestr, \mediumr, \worstr, group-separable$\}$.

\begin{corollary}\label{cor:sp-sc-gs_vd_npc}
  \problemPiVD is \NP-complete for every property~$\Pi \in \{$single-peaked, single-caved, group-separable$\}$.
\end{corollary}

\begin{proof}
 First, the profile constructed in the proof of \cref{prop:best-worst-medium-vd_npc} does not contain any three alternatives~$x$, $y$, and $z$ such that there is
 one voter with $x \pref y \pref z$ and
 one voter with $z \pref y \pref x$.
 Thus, the profile does not contain any \acon-configuration or any $\bar{\acon}$-configuration.

 Second, one can partition every set~$T$ of four alternatives 
 into two non-empty subsets~$T_1$ and $T_2$ such that $T_1 \pref T_2$ holds for each voter
 because at most three alternatives can correspond to the same edge and 
 all voters have the same ranking over the alternatives 
 that correspond to different edges.
 However, this is not possible in a \bcon-configuration.
 Thus, the profile does not contain any \acon-configuration, 
 $\bar{\acon}$-configuration, or \binconsistentconfig.
 
 As a consequence, the reduction in the proof of \cref{prop:best-worst-medium-vd_npc} 
 also works for \problemSpVD, \problemScVD, and \problemGsVD.
\end{proof}

Just as the result for the maverick deletion, 
the \NP-hardness result of \problemMediumAD also transfers to the group-separable case.
After deleting the alternatives corresponding to a vertex cover, 
the resulting profile from the proof of \cref{prop:best-worst-medium-ad_npc}
does not contain any \binconsistentconfig{s}.
Thus, the following holds.

\begin{corollary}\label{cor:gs_ad_npc}
  \problemGsAD is \NP-complete.
\end{corollary}

\begin{proof}
 The profile constructed in the proof of \cref{prop:best-worst-medium-ad_npc}
 may contain \bcon-configurations, but we show that destroying all
 \mediuminconsistentconfig{s}
 by deleting at most $k$ alternatives also destroys all \bcon-configurations.
 Consider the profile~$\ppp$ after the the deletion of the alternatives.
 Assume towards a contradiction that profile~$\ppp$ contains a \bcon-configuration{} which involves four alternatives~$w, x, y, z$ and two voters~$v, v'$ with preferences
 \begin{align*}
   \defvoter v \colon w \pref x\pref y \pref z ~\text{ and }~ \defvoter v' \colon x \pref z \pref w \pref y.
 \end{align*}
 Observe that a \bcon-configuration may contain at most one dummy alternative,
 because no two alternatives appear consecutively in both \preferenceorder{s} of
 a \bcon-configuration,
 but all dummy alternatives appear consecutively in all
 \preferenceorder{s} of the profile~$\ppp$.
 Furthermore, $w$, $y$, and $z$ are vertex alternatives
 since in $\ppp$, no voter prefers more than one vertex alternative to a dummy alternative.
 Then by the definition of \bcon-configurations, 
 we have that voter~$v$ ranks~$a_{j} \pref x \pref a_{j'} \pref a_{j''}$
 and voter~$v'$ ranks $x \pref a_{j''} \pref a_{j} \pref a_{j'}$ 
 with~$a_{j}, a_{j'}, a_{j''} \in A$.
 However, the preference order of voter~$v$ implies that~$j'<j''$
 and the preference order of voter~$v'$ implies that $j''<j'$---a contradiction.
 As consequence, the reduction in the proof of \cref{prop:best-worst-medium-ad_npc}
 also works for \problemGsVD.
\end{proof}

\medskip
In order to be group-separable, a preference profile must be \mediumr{} and
\bcon-restricted. 
As already shown in \cref{cor:sp-sc-gs_vd_npc} and in \cref{cor:gs_ad_npc}, 
deleting as few maverick voters (or alternatives) as possible to obtain the group-separable property is \NP-hard. 
Alternatively, we can also derive this intractability result from the following two theorems. 

\begin{theorem}\label{thm:bcon-vd_npc}
  \problemBconVD is \NP-complete.
\end{theorem}

\begin{proof}
  We reduce from \problemVC to show \NP-hardness.
  Let $(G=(U,E), k)$ denote a \problemVC instance with vertex set~$U=\{u_{1}, \ldots, u_r\}$ and edge set~$E=\{e_{1}, \ldots, e_s\}$;
  without loss of generality we assume that the input graph~$G$ is connected and
  that it has at least four vertices, that is,~$r\ge4$.
  The set of alternatives consists of four \emph{edge alternative}s~$a_j, b_j$, $c_j$, $d_j$ 
  for each edge~$e_j \in E$.
  For each vertex in $U$, we construct one voter.
  That is, we define~$A \coloneqq \{a_j,b_j,c_j,d_j \mid e_j \in E\}$ and $V \coloneqq \{v_i \mid u_i \in U\}$.
  In total, the number~$m$ of alternatives is~$4s$ and the number~$n$
  of voters is~$r$.

  Now we specify the \preferenceorder{} of each voter.
  Every voter prefers~$\{a_{j}, b_{j},$ $c_{j},$ $d_{j}\}$ to $\{a_{j'}$, $b_{j'}$, $c_{j'}, d_{j'}\}$ whenever $j < j'$.
  For each edge~$e_j$ with two incident vertices~$u_i$ and $u_{i'}$, $i<i'$,
  and for each non-incident vertex~$u_{i''} \notin e_j$,
  the following holds:
  \begin{alignat*}{3}
    &\defvoter v_i    &\colon &\;\;  a_j \pref b_j \pref c_j \pref d_j, \\
    &\defvoter v_{i'}  &\colon &\;\;  b_j \pref d_j \pref a_j \pref c_j, \\
    &\defvoter v_{i''}  &\colon &\;\;  d_j \pref a_j \pref b_j \pref c_j, 
  \end{alignat*}
  In this way, 
  any $\bcon$-configuration regarding alternatives~$a_j, b_j, c_j, d_j$
  must involve voters~$v_{i}$ and $v_{i'}$.
  The maximum number of voters to delete equals the maximum vertex cover size~$k$.
  This completes the construction.

  Clearly, the whole construction runs in polynomial time.
  It remains to show 
  that~$(G=(U, E), k)$ has a vertex cover
  of size at most~$k$ if and only if the constructed profile can be made
  \betar{} by deleting at most~$k$ voters.

  For the ``only if'' part, suppose that $U'\subseteq U$ with $|U'|\le k$ is a vertex cover.
  We show that after deleting the voters corresponding to the vertices in $U'$ 
  the resulting profile is \betar.
  Suppose for the sake of contradiction that the resulting profile is 
  \emph{not} \betar.
  That is, it still contains a \binconsistentconfig{}~$\sigma$.
  Since all voters prefer $\{a_j, b_j, c_j, d_j\}$ to $\{a_{j'}, b_{j'}, c_{j''}, d_{j'''}\}$ whenever $j<j'$,
  the profile restricted to every four alternatives that correspond to at least two edges,
  is group-separable.
  But since $\sigma$ is not group-separable,
  $\sigma$ involves four alternatives~$a_j$, $b_j$, $c_j$, and $d_j$ that correspond to a single edge~$e_j$.
  As already observed, any $\bcon$-configuration regarding alternatives~$a_j, b_j, c_j, d_j$
  must involve voters~$v_{i}$ and $v_{i'}$ that correspond to both endpoints of $e_j$.
  Then edge~$e_j$ is not covered by any vertex in~$U'$---a contradiction.

  For the ``if'' part, suppose that the profile becomes \betar{}
  by deleting a subset~$V'\subseteq V$ of voters with $|V'|\le k$.
  That is, no two voters form a \binconsistentconfig{}.
  We show by contradiction that~$V'$ corresponds to a vertex cover of graph~$G$ and has size at most $k$.
  Clearly, $|V'| \le k$.
  Assume towards a contradiction that an edge~$e_j$ is not covered
  by the vertices corresponding to the voters in $V'$.
  Then, the two voters corresponding to the vertices that are incident with edge~$e_j$ form a \binconsistentconfig{} with regard to $a_j$, $b_j$, $c_j$, and $d_j$---a contradiction.
  Thus, $V'$ corresponds to a vertex cover of graph~$G$ and its size is at most $k$.
\end{proof}

\begin{theorem}\label{thm:bcon-ad_npc}
  \problemBconAD is \NP-complete.
\end{theorem}

\begin{proof}
  We reduce from \problemVC to \problemBconAD.
  Let $(G=(U, E), k)$ denote a \problemVC instance with vertex set~$U=\{u_{1}, \ldots, u_{r}\}$ 
  and edge set~$E=\{e_{1}, \ldots, e_{s}\}$;
  without loss of generality we assume that the input graph~$G$ is connected and
  that $r\ge k+2$.
  The set of alternatives consists of one \emph{vertex alternative}~$a_j$ and one \emph{dummy alternative}~$d_j$ for each vertex~$u_j$ in $U$.
  Let $A$ denote the set of all vertex alternatives and
  let $D$ denote the set of all dummy alternatives.
  The number~$m$ of constructed alternatives is $2r$.

  We fix the canonical order of $A\cup D$ to be 
  \[\seq{A\cup D} \coloneqq d_1 \pref a_1 \pref d_2 \pref a_2 \pref \ldots \pref d_r \pref a_r.\]
  We introduce a voter~$v_{0}$ with the preference order~$\seq{A\cup D}$.
  For each edge~$e_i = \{u_{j}, u_{j'}\}$ with $j < j'$,
  we introduce one \emph{\edgevoter}~$v_{i}$ with \preferenceorder{}
  \[a_{j} \pref a_{j'} \pref \seq{A\cup D \setminus \{a_{j}, a_{j'}\}}.\]
  Observe that voter~$v_0$ and $v_{i}$ form a \binconsistentconfig{} with respect to the four alternatives corresponding to the vertices in edge~$\{u_{j}, u_{j'}\}$ with $j<j'$,
  that is, with respect to $a_j$, $a_{j'}$, $d_j$, and $d_{j'}$.
  \begin{alignat*}{3}
    &\defvoter v_0 &\colon &\;\; d_j \pref a_j \pref d_{j'} \pref a_{j'},\\
    &\defvoter v_i &\colon &\;\; a_j \pref a_{j'} \pref d_{j} \pref d_{j'}.  
  \end{alignat*}
  Furthermore, for each pair of alternatives~$a$ and $b$,
  if there is a voter~$v$ preferring~$a$ to $b$,
  then the following holds: 
  \begin{enumerate}[(i)]
    \item\label{prop:not-in-the-first-two} 
    If neither $a$ nor $b$ is in the first two positions of voter~$v$'s preference order,
    then $a$ and $b$ correspond to two (not necessarily distinct) vertices~$u_{j}$ and $u_{j'}$ with $j \le j'$.
    \item\label{prop:when-edge} If $a$ and $b$ correspond to two vertices~$u_j$ and $u_{j'}$ with $j > j'$
    and if there is a third alternative~$c$ such that $v$ prefers $c$ to $a$,
    then $c$ and $a$ correspond to two adjacent vertices.
  \end{enumerate}
  We will utilize these two facts several times to show some contradictions.

  Let $V$ denote the set of all voters.
  In total, the number~$n$ of constructed voters is $s+1$.
  The maximum number of alternatives to delete equals the maximum vertex cover size~$k$.
  This completes the construction.

  Our construction runs in polynomial time.  
  It remains to show 
  that $(G=(U, E), k)$ has a vertex cover of size at most~$k$ 
  if and only if the constructed profile can be made \betar{} by deleting at most $k$~alternatives. 

  For the ``only if'' part, suppose that $U'\subseteq U$ with $|U'|\le k$ is
  a vertex cover of graph~$G$.
  We show that after deleting the vertex alternatives corresponding to~$U'$, denoted by $A'$,
  the resulting profile is \betar{}.
  Suppose for the sake of contradiction that the resulting profile still 
  contains a \binconsistentconfig~$\sigma$ with regard to four alternatives~$w, x, y, z$ and two voters~$v, v'$ with preferences
  \begin{align*}
    \defvoter v\colon w\pref x \pref y \pref z ~\text{ and }~ \defvoter v'\colon x \pref z \pref w \pref y.
  \end{align*}
  Let $w, x, y, z$ correspond to four non-deleted vertices $u_j$, $u_{j'}$, $u_{j''}$, $u_{j'''}$, respectively.
  By Property~(\ref{prop:not-in-the-first-two}), 
  the preference order of voter~$v$ implies that $j''\le j'''$ (note that 
  voter~$v$ ranks neither $y$ nor $z$ in the first two positions)
  and the preference order of voter~$v'$ implies that $j \le j''$ (note that 
  voter~$v'$ ranks neither $w$ nor $y$ in the first two positions).
  Since $x$, $y$, and $z$ correspond to at least two distinct vertices,
  it follows that $j<j'''$.
  Since voter~$v'$ rank $x\pref z \pref w \pref y$,
  by Property~(\ref{prop:when-edge})
  the inequality $j< j'''$ implies that $u_{j'}$ and $u_{j'''}$ are adjacent---a contradiction to $U'$ being a vertex cover. 
  Indeed, the resulting profile is group-separable. To see this, note that
  any size-at-least-three subset~$T \subseteq (A\cup D) \setminus A'$ of alternatives can be
  partitioned into two non-empty subsets~$\{a\}$ and $T\setminus \{a\}$ 
  with~$a$ being the last alternative in the canonical order restricted to the alternatives in set~$T$.

  For ``if'' part, suppose that the constructed profile is a yes-instance of \problemBconAD.
  Let $A'\subseteq A\cup D$ be the set of deleted alternatives with $|A'|\le k$.
  Consider the vertex set~$U'$ containing all vertices corresponding to a vertex alternative
  or to a dummy alternative in $A'$,
  that is, $U' \coloneqq \{u_j \mid a_j \in A' \vee d_j \in A'\}$.
  Obviously, $|U'|\le k$.
%
  We show that set~$U'$ is a vertex cover.
  Suppose towards a contradiction that there is an uncovered edge~$e_i=\{u_j, u_{j'}\}$ with $j < j'$.
  By the definition of $U'$, we have that~$A'\cap\{d_j, a_j, d_{j'}, a_{j'}\}=\emptyset$.
  Then, voters~$v_0$ and~$v_i$ form a \binconsistentconfig{} with respect to the alternatives~$d_j, a_j, d_{j'}, a_{j'}$---a contradiction.
\end{proof}

\section{Single-crossing properties}
\label{sec:sc-results}
In this section, we show that for the single-crossing property,
the alternative deletion problem is \NP-hard while the maverick deletion problem is polynomial-time solvable.
The \NP-hardness proof is based on the following \NP-complete \textsc{Maximum 2-Satisfiability} problem~\cite{GJ79}.
 \decprob{Maximum 2-Satisfiability~(Max2Sat)}
  {A set~$U$ of Boolean variables, a collection~$C$ of size-two clauses over~$U$ and a positive integer~$h$.}
  {Is there a truth assignment for~$U$ such that at least~$h$ clauses in~$C$ are satisfied?}

\begin{theorem}\label{thm:sc_ad_npc}
  \problemScAD  is \NP-complete.
\end{theorem}

\begin{proof}\label{proof:sc_ad_npc}
  For the \NP-hardness result we reduce from \textsc{Max2Sat}~\cite{GJ79}. 
  %
  We will provide an example for the reduction (\cref{tab:sc-ad-reduced-instance}) right after this proof.

  Let $(U, C, h)$ be a \textsc{Max2Sat} instance with variable set~$U = \{x_{1},\ldots, x_r\}$ and clause set~$C = \{c_{1}, \ldots, c_s\}$.
  We construct two sets~$O$ and $\overline{O}$ of \emph{dummy alternatives} with
  $|O|=|\overline{O}|=2(r \cdot s+r+s)+1$ .
  For each variable~$x_i \in U$, we construct two sets~$X_i$  and $\overline{X_i}$ of
  \emph{variable alternative}s with $|X_i| = |\overline{X_i}| = s+1$.
  We say that $X_i$ corresponds to~$x_i$ and that $\overline{X_i}$ corresponds to~$\overline{x_i}$.
  The canonical orders $\seq{O}$, $\seq{\overline{O}}$,
  $\seq{X_i}$ and $\seq{\overline{X_i}}$, $i \in\{1,\ldots,r\}$, are arbitrary but fixed.
  Let $X$ be the union $\bigcup_{i=1}^{r}{X_i \cup \overline{X_i}}$ of all
  variable alternatives.
  The canonical order~$\seq{X}$ is defined as
  \begin{align*}
   \seq{X} \coloneqq \seq{X_{1}} \pref \seq{\overline{X_{1}}} \pref \seq{X_{2}} \pref \seq{\overline{X_{2}}} \pref \ldots \pref  \seq{X_r}  \pref 
    \seq{\overline{X_r}}.
  \end{align*}
  For each clause~$c_j \in C$,
  we construct two \emph{clause alternative}s $a_j$ and $b_j$.
  Let~$A$ denote the set of all clause alternatives.
  The canonical order~$\seq{A}$ is defined as
  \[\seq{A} \coloneqq a_{1} \pref b_{1} \pref a_{2} \pref b_{2} \pref \ldots \pref a_s \pref b_s.\]
  The total number~$m$ of alternatives is $6(r \cdot s+r+s)+2$.

  We will introduce voters and their \preferenceorder{s} such
  that 
  \begin{enumerate}[(1)]
    \item deleting all alternatives in $X_i$ corresponds to setting variable~$x_i$ to true, 
    \item deleting all alternatives in $\overline{X_i}$ corresponds
    to setting variable~$x_i$ to false, and
    \item deleting~$b_j$ or~$a_j$ corresponds to clause~$c_j$ not being
    satisfied.
  \end{enumerate}
  We construct two sets~$V$ and $W$ of voters with $|V| = 2r$ and $|W|=4s$.
  Voter set~$V$ consists of two voters~$v_{2i-1}$ and $v_{2i}$ for
  each variable~$x_{i}$, $1 \le i \le r$. 
  Their \preferenceorder{s} are
\newcommand{\sldots}{\ldots}
\newcommand{\sseq}[1]{\langle#1\rangle}
\newcommand{\bboldsymbol}{}
\newcommand{\prefdotssucc}{\pref \ldots \pref}
    \begin{alignat*}{3}
      &\seq{O} \pref \seq{\overline{O}} \pref\; && \seq{X_1} \pref \seq{\overline{X_{1}}} \prefdotssucc \seq{X_{i-1}} \pref \seq{\overline{X_{i-1}}} \pref\\
      &&& \seq{\overline{X_{i}}} \pref \seq{X_i} \pref \seq{X_{i+1}} \pref \seq{\overline{X_{i+1}}} \prefdotssucc \seq{X_{r}} \pref \seq{\overline{X_{r}}} \pref \sseq{A},  \\
      & \seq{\overline{O}} \pref \seq{O} \pref && \seq{X_1} \pref \seq{\overline{X_{1}}} 
      \prefdotssucc 
      \seq{X_{i-1}} \pref \seq{\overline{X_{i-1}}} 
      \pref\\
      &&& \seq{\overline{X_{i}}} \pref \seq{X_i} \pref 
      \seq{X_{i+1}} \pref \seq{\overline{X_{i+1}}} \prefdotssucc 
      \seq{X_{r}} \pref \seq{\overline{X_{r}}} \pref \sseq{A}
      \text{,} 
    \end{alignat*}
respectively.
These two voters together with any other two voters~$v_{\ell}$ and $v_{\ell'}\in V\setminus\{v_{2i-1},
v_{2i}\}$ with odd number~$\ell$ and even number~$\ell'$ form a \dinconsistentconfig{} with regard to
each four alternatives~$o, \overline{o}, x, \overline{x}$ such that
$o \in O, \overline{o}\in \overline{O}, x \in X_i, \overline{x} \in \overline{X_i}$:
  \begin{alignat*}{3}
    &\defvoter v_{2i-1} &\colon\;\;& o \pref \overline{o} ~\text{ and }~ \overline{x} \pref x, \\
    &\defvoter v_{2i} &\colon\;\;& \overline{o} \pref o ~\text{ and }~ \overline{x} \pref x,\\
    &\defvoter v_{\ell} &\colon\;\;& o \pref \overline{o} ~\text{ and }~ x \pref \overline{x},\\
    &\defvoter v_{\ell'} &\colon\;\;& \overline{o} \pref o ~\text{ and }~ x \pref \overline{x}.
  \end{alignat*}
  Voter set~$W$ consists of four voters~$w_{4j-3}, w_{4j-2}, w_{4j-1}, w_{4j}$ for each clause~$c_j$, $1 \le j \le s$.  
  These four voters
  have the same \preferenceorder
  \begin{align*}
   \seq{\overline{O}} \pref \seq{O} \pref \seq{A_{1}} \pref \seq{X} \pref \seq{A_2}
  \end{align*} over the set~$O\cup \overline{O} \cup A_{1} \cup A_2 \cup X$,
  where $A_{1}=\{a_{j'}, b_{j'} \mid j' < j\}$ and $A_2=\{a_{j'}, b_{j'} \mid
  j' > j\}$.
  Note that $A_1\cup A_2 = A \setminus \{a_j, b_j\}$.
  Thus, it remains to specify the exact positions of $a_j$ and $b_j$ in the four voters' \preferenceorder{s}:
  Let \x{j}{1} denote the set of
  variable alternatives corresponding to the literal in~$c_j$ with the lower index
  and \x{j}{2} denote the set of variable alternatives corresponding to the
  literal in~$c_j$ with the higher index.
  For instance, if $c_j=\overline{x_2} \vee x_{4}$, then \x{j}{1}
  equals~$\overline{X_2}$ and \x{j}{2} equals $\overline{X_4}$.

  Voters~$w_{4j-3}$ and $w_{4j-2}$ rank the clause alternative~$a_j$ right
  below the last alternative in $\seq{\x{j}{1}}$ while voters~$w_{4j-1}$
  and $w_{4j}$ rank it right above the first alternative in $\seq{\x{j}{1}}$.
  As for alternative~$b_j$, voters~$w_{4j-3}$ and $w_{4j-1}$ rank~$b_j$ right above
  the first variable alternative in $\seq{\x{j}{2}}$ while
  voters~$w_{4j-2}$ and $w_{4j}$ rank it right below the last variable
  alternative in $\seq{\x{j}{2}}$.
  Thus, these four voters form a
  \dinconsistentconfig{} with regard to $a_{j}, b_{j}, x\in \x{j}{1}$, and $y\in \x{j}{2}$:
  \begin{alignat*}{3}  
    &\defvoter w\sub{4j-3}&\colon\;\;&x \pref a\sub{j} ~\text{ and }~ b\sub{j} \pref {y}, \\
    &\defvoter w\sub{4j-2}&\colon\;\;& x \pref a\sub{j} ~\text{ and }~ {y} \pref b\sub{j},\\
    &\defvoter w\sub{4j-1}&\colon\;\;& a\sub{j} \pref x ~\text{ and }~ b\sub{j} \pref {y}, \\
    &\defvoter w\sub{4j}&\colon\;\;& a\sub{j} \pref x  ~\text{ and }~ {y} \pref b\sub{j}.
  \end{alignat*}

  \noindent
  We complete the construction by setting the number~$k$ of alternatives that may be deleted to $k\coloneqq r(s+1)+(s-h)$.

  The construction clearly runs in polynomial time.
  It remains to show that $(U, C, h)$ is a yes-instance of \textsc{Max2Sat}
  if and only if the constructed profile together with $k$ is a yes-instance
  of \problemScAD.

  For the ``only if'' part, suppose that 
  there is a truth assignment $U  \to  \{\truevalue, \falsevalue\}^{r}$ of the variables 
  such that at least $h$ clauses are satisfied.
  We delete all variable alternatives in $X_i$ if $x_i$ is assigned to true.
  Otherwise, we delete all variable alternatives in $\overline{X_i}$.
  Furthermore, we delete the clause alternative~$b_j$ if $c_j$ is not satisfied by the assignment.
  Let $X_{\mathrm{rem}}$ be the set of remaining variable alternatives, 
  and let $A_{\mathrm{rem}}$ be the set of remaining clause alternatives.
  Then, $|X_{\mathrm{rem}}|=r(s+1)$ and $|A'|\ge s+h$,
  implying that the number of deleted alternatives is $|X|+|A| - (|X_{\mathrm{rem}}| + |A_{\mathrm{rem}}|) \le r(s+1) + (s-h) = k$. 
  
  For each $j\in\{1,\ldots,s\}$, 
  we define $\seq{z_{j}} = w\sub{4j-2} \pref w\sub{4j} \pref w\sub{4j-3} \pref w\sub{4j-1} $ 
  if the literal in clause~$c_j$ with the lower index is satisfied; 
  otherwise, $\seq{z_{j}} = w\sub{4j-3} \pref w\sub{4j-2} \pref w\sub{4j-1} \pref w\sub{4j}$.  
  The resulting profile is
  single-crossing with respect to the voter order~$L$:
  \[
  v_{1} \pref v_{3} \pref \ldots \pref v_{2r-1} \pref
  v_{2} \pref v_{4} \pref \ldots \pref v_{2r} \pref 
  \seq{z_{1}} \pref \seq{z_{2}} \pref \ldots \pref \seq{z_{s}}\text{.}
  \]

  Suppose for the sake of contradiction that~$L$ is not a single-crossing order,
  which means that
  there is a pair~$\{a, a'\} \subset O\cup \overline{O} \cup X_{\mathrm{rem}}\cup A_{\mathrm{rem}}$ of alternatives
  and three voters~$u,v,w$ with $u\pref_L v \pref_L w$ such that
  voter $v$ disagrees with voters~$u$ and $w$ on the relative order of $a$ and $a'$.

  Note that all voters along $L$ up to and including voter~$v_{2r-1}$ rank $\seq{O} \pref \seq{\overline{O}} \pref \seq{X}$ while all voters from $v_{2}$~onwards rank $\seq{\overline{O}} \pref \seq{O} \pref \seq{X}$.
  Hence,
  $a$ and~$a'$ can neither both be in $O\cup \overline{O}$, 
  nor both be in $X_{\mathrm{rem}}$.  
  Furthermore, $a$ and~$a'$ cannot both be in $A_{\mathrm{rem}}$ 
  as all voters rank~$\seq{A}$.  
  Moreover, since all voters rank $(O\cup
  \overline{O}) \pref (X\cup A)$, neither $a$ nor $a'$ belongs to $O\cup
  \overline{O}$.  This means, without loss of generality, that $a\in X_{\mathrm{rem}}$ and $a' \in
  A_{\mathrm{rem}}$.
  
  Assume that~$a'$ corresponds to clause~$c_j$ for some $j$, that is, $a'\in \{a_j, b_j\}$. 
  Then, for each alternative~$a''\in X_{\mathrm{rem}}\setminus (\x{j}{1} \cup \x{j}{2})$
  that does not correspond to a literal in $c_j$~(recall that $\x{j}{1}$ 
  and $\x{j}{2}$ denote the two sets of variable alternatives corresponding to the literal in~$c_j$ 
  with the lower index and the literal in~$c_j$ with the lower index, respectively),
  the following holds.
  If the first voter in $\seq{z_j}$ prefer~$a''$ to $a'$, 
  which means either that $a''$ is ranked in front of $\x{j}{1} \cup \x{j}{2}$ (by all voters)
  or 
  that $a'=a_j$ and $a''$ is ranked in front of $\x{j}{2}$ (by all voters),
  then all voters along the order~$L$ up to and including the last voter in $\seq{z_j}$ prefer $a''$ to $a'$ 
  while all remaining voters prefer $a'$ to $a''$;
  otherwise all voters along the order~$L$ up to and including the last voter in $\seq{z_{j-1}}$ prefer $a''$ to $a'$
  while all remaining voters prefer $a'$ to $a''$.
  Thus, $a'$ cannot be in $X_{\mathrm{rem}} \setminus (\x{j}{1} \cup \x{j}{2})$.
  That is, we have $a \in \x{j}{1} \cup \x{j}{2}$.
  We distinguish four cases regarding $a$ and $a'$.
  \begin{enumerate}[(i)]
    \item If $a \in \x{j}{1}$ and if $a' = a_j$, 
    then the literal corresponding to $\x{j}{1}$ is not satisfied because $\x{j}{1}$ is not deleted.
    Thus, $\seq{z_j}$ is defined as $w_{4j-3} \pref w_{4j-2} \pref $ $w_{4j-1} \pref w_{4j}$.
    All voters along $L$ up to and including $w_{4j-2}$ prefer $a$ to~$a'$, 
    and all remaining voters prefer $a'$ to~$a$.
    \item If $a \in \x{j}{1}$ and if $a' = b_j$,
    then all voters along $L$ up to and including the last voter in~$\seq{z_j}$ prefer $a$ to $a'$,
    and all remaining voters prefer $a'$ to~$a$.
    \item If $a \in \x{j}{2}$ and if $a' = a_j$,
    then all voters along $L$ up to and including the last voter in~$\seq{z_{j-1}}$ prefer $a$ to $a'$,
    and all remaining voters $a'$ to $a$.
    \item If $a \in \x{j}{2}$ and if $a' = b_j$, 
    then clause~$c_j$ is satisfied because $b_j$ is not deleted.
    Furthermore, since $\x{j}{2}$ is not deleted, $\x{j}{1}$ must be deleted because clause~$c_j$ is satisfied.
    This implies that the literal in clause~$c_j$ with the lower index is satisfied.
    Thus, $\seq{z_j}$ is defined as $w_{4j-2} \pref w_{4j} \pref w_{4j-3} \pref w_{4j-1}$.
    All voters along $L$ up to and including $w_{4j}$ prefer $a$ to~$a'$, 
    and all remaining voters prefer $a'$ to $a$.
  \end{enumerate}
  In summary, there is single a voter~$v$ along the order~$L$ such that
  all voters up to and including $v$ have the same preference over $\{a,a'\}$ and
  all remaining voters have the same preference over $\{a,a'\}$---a contradiction to the assumption that~$L$ is not a single-crossing order.

 \begin{table}[t!]
   \centering
   \begin{tabular}{@{}l @{} l @{}}
     $\defvoter v_1 \colon$&$ \seq{O} \spref \seq{\overline{O}} \spref \seq{\overline{X_1}} \spref \seq{X_1} \spref \seq{X_2} \spref \seq{\overline{X_2}} \spref a_1 \spref b_1 \spref a_2 \spref b_2 \spref a_3 \spref b_3 \spref a_4 \spref b_4$\\
     $\defvoter v_2 \colon$&$ \seq{\overline{O}} \spref \seq{O} \spref \seq{\overline{X_1}} \spref \seq{X_1} \spref \seq{X_2} \spref \seq{\overline{X_2}} \spref a_1 \spref b_1 \spref a_2 \spref b_2 \spref a_3 \spref b_3 \spref a_4 \spref b_4$\\
     $\defvoter v_3 \colon$&$ \seq{O} \spref \seq{\overline{O}} \spref \seq{X_1} \spref \seq{\overline{X_1}} \spref \seq{\overline{X_2}} \spref \seq{X_2} \spref a_1 \spref b_1 \spref a_2 \spref b_2 \spref a_3 \spref b_3 \spref a_4 \spref b_4$\\
     $\defvoter v_4 \colon$&$ \seq{\overline{O}} \spref \seq{O} \spref \seq{X_1} \spref \seq{\overline{X_1}} \spref \seq{\overline{X_2}} \spref \seq{X_2} \spref a_1 \spref b_1 \spref a_2 \spref b_2 \spref a_3 \spref b_3 \spref a_4 \spref b_4$\\
     $\defvoter w_1 \colon$&$ \seq{\overline{O}} \spref \seq{O} \spref \seq{X_1} \spref a_1 \spref \seq{\overline{X_1}} \spref b_1 \spref \seq{X_2} \spref \seq{\overline{X_2}} \spref a_2 \spref b_2 \spref a_3 \spref b_3 \spref a_4 \spref b_4$\\
     $\defvoter w_2 \colon$&$ \seq{\overline{O}} \spref \seq{O} \spref \seq{X_1} \spref a_1 \spref \seq{\overline{X_1}} \spref \seq{X_2} \spref b_1 \spref \seq{\overline{X_2}} \spref a_2 \spref b_2 \spref a_3 \spref b_3 \spref a_4 \spref b_4$\\
     $\defvoter w_3 \colon$&$ \seq{\overline{O}} \spref \seq{O} \spref a_1 \spref \seq{X_1} \spref \seq{\overline{X_1}} \spref b_1 \spref \seq{X_2} \spref \seq{\overline{X_2}} \spref a_2 \spref b_2 \spref a_3 \spref b_3 \spref a_4 \spref b_4$\\
     $\defvoter w_4 \colon$&$ \seq{\overline{O}} \spref \seq{O} \spref a_1 \spref \seq{X_1} \spref \seq{\overline{X_1}} \spref \seq{X_2} \spref b_1 \spref \seq{\overline{X_2}} \spref a_2 \spref b_2 \spref a_3 \spref b_3 \spref a_4 \spref b_4$\\
     $\defvoter w_5 \colon$&$ \seq{\overline{O}} \spref \seq{O} \spref a_1 \spref b_1 \spref \seq{X_1} \spref a_2 \spref \seq{\overline{X_1}} \spref \seq{X_2} \spref b_2 \spref \seq{\overline{X_2}} \spref a_3 \spref b_3 \spref a_4 \spref b_4$\\
     $\defvoter w_6 \colon$&$ \seq{\overline{O}} \spref \seq{O} \spref a_1 \spref b_1 \spref \seq{X_1} \spref a_2 \spref \seq{\overline{X_1}} \spref \seq{X_2} \spref \seq{\overline{X_2}} \spref b_2 \spref a_3 \spref b_3 \spref a_4 \spref b_4$\\
     $\defvoter w_7 \colon$&$ \seq{\overline{O}} \spref \seq{O} \spref a_1 \spref b_1 \spref a_2 \spref \seq{X_1} \spref \seq{\overline{X_1}} \spref \seq{X_2} \spref b_2 \spref \seq{\overline{X_2}} \spref a_3 \spref b_3 \spref a_4 \spref b_4$\\
     $\defvoter w_8 \colon$&$ \seq{\overline{O}} \spref \seq{O} \spref a_1 \spref b_1 \spref a_2 \spref \seq{X_1} \spref \seq{\overline{X_1}} \spref \seq{X_2} \spref \seq{\overline{X_2}}  \spref b_2  \spref a_3 \spref b_3 \spref a_4 \spref b_4$\\
     $\defvoter w_9 \colon$&$ \seq{\overline{O}} \spref \seq{O} \spref a_1 \spref b_1 \spref a_2 \spref b_2 \spref \seq{X_1} \spref \seq{\overline{X_1}} \spref a_3 \spref b_3 \spref \seq{X_2} \spref \seq{\overline{X_2}} \spref a_4 \spref b_4$\\
     $\defvoter w_{10} \colon$&$ \seq{\overline{O}} \spref \seq{O} \spref a_1 \spref b_1 \spref a_2 \spref b_2 \spref \seq{X_1} \spref \seq{\overline{X_1}} \spref a_3 \spref \seq{X_2} \spref b_3 \spref \seq{\overline{X_2}} \spref a_4 \spref b_4$\\
     $\defvoter w_{11} \colon$&$ \seq{\overline{O}} \spref \seq{O} \spref a_1 \spref b_1 \spref a_2 \spref b_2 \spref \seq{X_1} \spref a_3 \spref \seq{\overline{X_1}} \spref b_3 \spref \seq{X_2} \spref \seq{\overline{X_2}} \spref a_4 \spref b_4$\\
     $\defvoter w_{12} \colon$&$ \seq{\overline{O}} \spref \seq{O} \spref a_1 \spref b_1 \spref a_2 \spref b_2 \spref \seq{X_1} \spref a_3 \spref \seq{\overline{X_1}} \spref \seq{X_2} \spref b_3 \spref \seq{\overline{X_2}} \spref a_4 \spref b_4$\\
     $\defvoter w_{13}\colon$&$ \seq{\overline{O}} \spref \seq{O} \spref a_1 \spref b_1 \spref a_2 \spref b_2 \spref a_3 \spref b_3 \spref \seq{X_1} \spref \seq{\overline{X_1}} \spref a_4  \spref \seq{X_2} \spref b_4  \spref \seq{\overline{X_2}}$\\
     $\defvoter w_{14} \colon$&$ \seq{\overline{O}} \spref \seq{O} \spref a_1 \spref b_1 \spref a_2 \spref b_2 \spref a_3 \spref b_3 \spref \seq{X_1} \spref \seq{\overline{X_1}} \spref a_4 \spref \seq{X_2}  \spref \seq{\overline{X_2}}\spref b_4$\\
     $\defvoter w_{15} \colon$&$ \seq{\overline{O}} \spref \seq{O} \spref a_1 \spref b_1 \spref a_2 \spref b_2 \spref a_3 \spref b_3 \spref \seq{X_1} \spref a_4 \spref \seq{\overline{X_1}} \spref \seq{X_2}  \spref b_4 \spref \seq{\overline{X_2}}$\\
     $\defvoter w_{16} \colon$&$ \seq{\overline{O}} \spref \seq{O} \spref a_1 \spref b_1 \spref a_2 \spref b_2 \spref a_3 \spref b_3 \spref \seq{X_1} \spref a_4 \spref \seq{\overline{X_1}} \spref \seq{X_2} \spref \seq{\overline{X_2}}  \spref b_4$
   \end{tabular}
   \caption{An instance~$((A, V), k=11)$ with alternative set~$O \cup \overline{O} \cup X_1 \cup \overline{X_1} \cup X_2 \cup \overline{X_2} \cup \{a_i, b_i \mid 1\le i \le 4\}$ and voter set~$\{v_i, w_{4i-3}, w_{4i-2}, w_{4i-1}, w_{4i} \mid 1\le i \le 4\}$ reduced from the \problemMTwoSATshort instance with two variables~$x_1$ and $x_2$,
     and with four clauses~$c_1=x_1\wedge x_2$, $c_2=x_1 \wedge \overline{x_2}$, $c_3=\overline{x_1}\wedge x_2$, and $c_4=\overline{x_1} \wedge \overline{x_2}$.
 The maximum number~$h$ of satisfied clauses is three.
}
\label{tab:sc-ad-reduced-instance}
 \end{table}
 
  For the ``if'' part, suppose that deleting a set~$K$ of at most $k$ alternatives
  makes the remaining profile single-crossing.
  Since $|O|=|\overline{O}| \ge k$, at least one pair~$\{o, \overline{o}\}$ of dummy alternatives is not deleted, where $o\in O$ and $\overline{o} \in \overline{O}$.
  Let $X_{\mathrm{del}}$ denote the set of all deleted variable alternatives,
  and $A_{\mathrm{del}}$ denote the set of all deleted clause alternatives.
  Clearly, $|X_\mathrm{del}| + |A_{\mathrm{del}}| \le |K|$.
  For each $x_i \in U$, at least one set of $X_i$ and $\overline{X_{i}}$ must be deleted to destroy all \dcon-configurations involving alternatives in $\{o, \overline{o}\} \cup X_i \cup \overline{X_i}$.
  This means that $|X_{\mathrm{del}}| \ge r(s+1)$.
  Thus, $|A_{\mathrm{del}}| \le |K| - |X_{\mathrm{del}}| \le k - r(s+1) \le s-h$.
  Let $C_{\mathrm{both}}$ denote the set of clauses such that 
  neither $a_j$ nor $b_j$ is deleted, 
  $1\le j\le s$, that is, $C_{\mathrm{both}} \coloneqq \{c_j \mid \{a_j, b_j\} \cap A_{\mathrm{del}} = \emptyset\}$.
  Set~$C_{\mathrm{both}}$ has cardinality at least $h$ because $|A_{\mathrm{del}}| \le s-h$.
  We show that by setting variable~$x_i \in U$ to true if $X_i \subseteq X_{\mathrm{del}}$, and false
  otherwise, all clauses~$c_{j}$ from $C_{\mathrm{both}}$ are satisfied.
  Suppose for the sake of contradiction that clause~$c_{j}\in C_{\mathrm{both}}$ 
  is not satisfied.
  This means that  $\{a_{j}, b_{j}\}\cap A_{\mathrm{del}} =\emptyset$, 
  and that both $\x{j}{1}$ and $\x{j}{2}$ are not completely contained in $X_{\mathrm{del}}$.
  But then, voters  $w_{4i-3}, w_{4i-2}, w_{4i-1}$, and $w_{4i}$ form a
  \dinconsistentconfig{} with regard to $a_j, b_j, x, x'$ with $x \in \x{j}{1} \setminus X_{\mathrm{del}}$ and $x'\in \x{j}{2} \setminus X_{\mathrm{del}}$
  ---a contradiction.
 \end{proof}

 We illustrate our \NP-hardness reduction through an example. 
 Consider a \problemMTwoSATshort~instance with two variables~$x_1$ and $x_2$ and four clauses~$c_1=x_1\wedge x_2$, 
 $c_2=x_1 \wedge \overline{x_2}$, $c_3=\overline{x_1}\wedge x_2$, and $c_4=\overline{x_1} \wedge \overline{x_2}$.
 The maximum number~$h$ of satisfied clauses is three.
 For instance, the truth assignment~$x_1 \mapsto \truevalue$ and $x_2 \mapsto \falsevalue$ satisfies clauses~$c_1, c_2, c_4$.
 \cref{tab:sc-ad-reduced-instance} depicts the reduced instance of \problemScVD
 with alternative set~$O \cup \overline{O} \cup X_1 \cup \overline{X_1} \cup X_2 \cup \overline{X_2} \cup \{a_i, b_i \mid 1\le i \le 4\}$ and voter set~$\{v_i, w_{4i-3}, w_{4i-2}, w_{4i-1}, w_{4i} \mid 1\le i \le 4\}$.
 The number~$k$ of alternatives that may be deleted is set to $2\cdot (4+1) + (4-3) = 11$.
 We can verify that deleting the alternatives from $X_1\cup \overline{X_2} \cup \{b_3\}$
 results in a single-crossing profile where a single-crossing order is $v_1 \pref v_3 \pref v_2 \pref v_4 \pref w_2 \pref w_4 \pref w_1 \pref w_3 \pref w_6 \pref w_8 \pref w_5 \pref w_7 \pref w_9 \pref w_{10} \pref \ldots \pref w_{16}$.

 \medskip

 In contrast to the other \NP-hard \problemPiVD problems,
 \problemScVD is polynomial-time solvable.
 The algorithm, which is similar to the single-crossing detection algorithm by~\citet{EFS12},
 not only solves the decision problem, 
 but also the optimization problem asking for the maximum-size subset of voters
 such that the profile restricted to this subset is single-crossing.

 Before we proceed to describe the algorithm, we define some notions and make some observations about single-crossing profiles.
 We call a set~$S$ of preference orders \emph{single-crossing}, 
 if there is a single-crossing order of the elements in $S$.
 We introduce the notion~\mbox{$\Delta(\pref,\pref')$} of the set of \emph{conflict pairs} for two given \preferenceorder{s}~$\pref$ and $\pref'\in \mathcal{S}$.
 By $\Delta(\pref,\pref')$, 
 we denote the set of pairs~$\{a, b\}$ of alternatives whose relative order differs between \preferenceorder{s}~$\pref$ and $\pref'$.
 Formally, $\Delta(\pref,\pref') \coloneqq \{\{a,b\} \mid a\pref b \wedge b \pref' a\}.$
 For instance, given three alternatives~$a,b,c$, 
 if the \preferenceorder{s}~$\pref, \pref'$ are the same,
 then $\Delta(v,v') = \emptyset$;
 if the two \preferenceorder{s} are specified as follows:
 \[
   b \pref a \pref c \text{\; and \;} c \pref' b \pref' a\text{,}
 \]
 then $\Delta(\pref,\pref') \coloneqq \{\{a,c\}, \{b,c\}\}$.

 Based on this notion, 
 we can redefine the single-crossing property of a set of \preferenceorder{s} 
 using set inclusions.
 For the sake of readability, 
 we will use the vector notation~$(\cdot, \cdots, \cdot)$ to denote a linear order over a set of preference orders. 

 \begin{lemma}\label{lem:single-crossing<=>set inclusion}
   A linear order~$(\pref^*_1, \pref^*_2, \ldots, \pref^*_n)$ over a set of $n$ \preferenceorder{s} is single-crossing if and only if
   for each two preference orders~$\pref^*_i$ and $\pref^*_j$ with $1\le i\le j \le n$ it holds that
   $\Delta(\pref^*_1, \pref^*_i) \subseteq \Delta(\pref^*_1, \pref^*_j)$.
 \end{lemma}

 \begin{proof}
   The ``only if'' part follows directly from the definition of the single-crossing property and the set of conflict pairs.
   For the ``if'' part, suppose towards a contradiction that the order~$(\pref^*_1, \pref^*_2, \ldots, \pref^*_n)$ is not single-crossing.
   This means that there are two alternatives~$a, b$,
   and there are two preference orders~$\pref^*_i, \pref^*_j$ with $1 < i < j$ such that $a\pref^*_{1} b$ and $a\pref^*_{j} b$, but $b\pref^*_{i} a$.
   Then it follows that $\{a,b\}\in \Delta(\pref^*_{1}, \pref^*_{i})$ but $\{a, b\} \notin \Delta(\pref^*_1, \pref^*_j)$---a contradiction.
 \end{proof}
 
 The following observation states that 
 the single-crossing property only depends on the preference orders,
 not on the voters.
 \begin{obs}\label{obs:identicalvoters do not change sc property}
   Let $V$ be a set of voters and let $w\notin V$ be an additional voter 
   such that there is a voter in $V$ who has the same \preferenceorder{} as voter~$w$.
   Then, the profile with voter set~$V$ is single-crossing if and only if
   the profile with voter set~$V\cup \{w\}$ is single-crossing.
 \end{obs}
 \begin{proof}
   By the definition of single-crossing orders,
   a profile is single-crossing if and only if the set of the preference orders of all voters in this profile is single-crossing.
   Since adding voter~$w$ to voter set~$V$ does not change the set of the preference orders of all voters in $V$,
   the statement follows.
 \end{proof}

 Based on the notions of conflicting pairs and single-crossing sets of preference orders,
 and \cref{lem:single-crossing<=>set inclusion} and \cref{obs:identicalvoters do not change sc property},
 we can solve the maximization variant of the \problemScVD problem by 
 reducing it to finding a longest path in an appropriately constructed directed acyclic graph.
 This implies the following theorem.

\begin{theorem}\label{thm:sc_vd_p}
  \problemScVD is solvable in $O(n^3 \cdot m^2)$ time, where $n$~denotes the
  number of voters and $m$~denotes the number of alternatives.
\end{theorem}

\begin{proof}\label{proof:sc_vd_p}
  Suppose that we are given a profile with $A$ being the set of
  $m$~alternatives and $V$ being the set of $n$~voters, each voter having a preference order over $A$.
  Now, the goal is to find a maximum-size subset of voters such that the profile restricted to this subset is single-crossing. 
  To achieve this goal, we use two further notions:
  Let $\mathcal{S}(V) \coloneqq \{\pref_v \mid v \in V\}$ be the set of the preference orders of all voters from $V$.
  Without loss of generality, let $\mathcal{S}(V) \coloneqq \{\pref_1, \pref_2, \ldots, \pref_{n'}\}$.
  For each \preferenceorder~$\pref\in \mathcal{S}(V)$, 
  let $\#(\pref, V)$ denote the number of voters in $V$ with the same \preferenceorder~$\pref$.
  By \cref{obs:identicalvoters do not change sc property},
  it follows that finding the maximum-size single-crossing voter subset is equivalent to 
  finding a \emph{single-crossing} subset~$\mathcal{S}'\subseteq \mathcal{S}(V)$ of \preferenceorder{s} that maximizes the sum~$\sum_{\pref\in \mathcal{S'}}{\#(\pref, V)}$.

  Now, observe that if $\pref$ is the first preference order along the single-crossing order over set~$\mathcal{S}'$,
  then for each two further preference orders~$\pref', \pref''\in \mathcal{S'}$ 
  with $\pref'$ being the predecessor of $\pref''$ along the single-crossing order, 
  by \cref{lem:single-crossing<=>set inclusion},
  it holds that $\Delta(\pref,\pref') \subseteq \Delta(\pref,\pref'')$.
  This inspires us to build a directed graph based on the set inclusion relation and then, to find a maximum-weight path.
  Thus, the idea of our algorithm is to first construct a directed graph with weighted arcs 
  and then to find a maximum-weight path on this graph.
  We will provide an example to illustrate this idea right after this proof.

  The construction of the desired directed graph works as follows:
  For each two numbers~$z,i \in \{1, 2, \ldots, n'\}$,
  we construct one vertex~$u^z_i$; this vertex will represent the \preferenceorder~$\pref_i$ in a linear order starting with \preferenceorder~$\pref_z$.
  Then, for each further number~$i' \in \{1,2,\ldots, n'\}$ with $i\neq i'$, 
  we add an arc with weight $\#(\pref_{i'}, V)$ from vertex~$v^z_i$ to vertex~$v^z_{i'}$
  if $\Delta(\pref_z, \pref_i) \subseteq \Delta(\pref_z,\pref_{i'})$.
  Finally, we construct a root vertex~$u_r$, 
  and for each number $z\in \{1,2,\ldots, n'\}$, we add an arc with weight~$\#(\pref_{z}, V)$ from root~$u_r$ to~$u^z_z$.
  This completes the construction.
  Observe that the constructed directed graph is acyclic:
  \begin{enumerate}
    \item For each three numbers~$z, z', i \in \{1,2,\ldots, n'\}$ with $z\neq z'$,
    there are no arcs between vertices~$u^z_{i}$ and $u^{z'}_{i}$.
    \item For each three numbers~$z, i, i' \in \{1,2,\ldots, n'\}$ with $i\neq i'$, 
    a path from $u^{z}_i$ to $u^{z}_{i'}$ implies that $\Delta(\pref_z, \pref_i)\subseteq \Delta(\pref_z,\pref_{i'})$,
    while a path from $u^{z}_{i'}$ to $u^{z}_{i}$ implies that $\Delta(\pref_z, \pref_{i'})\subseteq \Delta(\pref_z,\pref_{i})$.
    Thus, both paths cannot exist simultaneously because $\pref_{i} \neq \pref_{i'}$.
  \end{enumerate}
  
  Now, an order of the vertices along a maximum-weight directed path corresponds to a subset~$\mathcal{S}'\subseteq \mathcal{S}(V)$ of \preferenceorder{s}, 
  such that $\mathcal{S}'$ is \emph{single-crossing},
  and the sum~$\sum_{\pref\in \mathcal{S'}}{\#(\pref, V)}$ is maximum:
  The second vertex on the maximum-weight path fixes the first \preferenceorder{} of
  the single-crossing order.
  Each successive vertex~$u^z_i$ on the path represents the successive \preferenceorder~$\pref_i$
  in the single-crossing order (this is true by \cref{lem:single-crossing<=>set inclusion} and by the way we define an arc).
  The arc weights ensure that the sum of the weights on the path equals the total number of represented voters.
  
  As to the running time analysis, we need $O(n\cdot m)$ time to compute the set~$\mathcal{S}(V)$.
  Then, for each two (not necessarily distinct) \preferenceorder{s}~$\pref, \pref' \in \mathcal{S}(V)$, we compute $\Delta(\pref, \pref')$.
  This can be done by checking the relative order of each pair of alternatives in $O(n^2 \cdot m^2)$ time.
  Further, we construct the directed graph in $O(n^3\cdot m^2)$~time.
  Finally, we compute the maximum-weight path in a directed acyclic graph with $n^2$~vertices and
  $n^3$~arcs in $O(n^3)$ time.
  To achieve this, 
  we first replace all positive weights~$w$ with $-w$, 
  and then use the algorithm in the textbook of \citet[Sec 24.2]{CorLeiRivSte2009}
  to find a minimum-weight path.
  In total, the running time is $O(n^3 \cdot m^2)$.
\end{proof}

Consider a profile~$\ppp$ with three alternatives~$a, b, c$ and four voters~$v_1, v_2, v_3, v_4$ whose preference orders are depicted in \cref{fig:sc-vd-example}(\subref{fig:sc-vd-example-sc-profile}).
This profile is not single-crossing since it contains a \ccon-configuration with regard to the alternatives~$a,c,a,$ $b,a,b$ and voters~$v_1,v_3,v_4$.
The set of the preference orders of all voters is $\{a\pref b \pref c,\; b\pref c \pref a,\; c\pref a\pref b\}$.
According to our algorithm of finding a single-crossing profile with maximum number of voters,
we first construct a weighted directed graph as depicted in \cref{fig:sc-vd-example}(\subref{fig:sc-vd-example-dag}).
Then, we will find a maximum-weight path in the graph.
We can verify that there are four maximum-weight paths,
including this one $u_r \to (a\pref b \pref c)\to (b\pref c \pref a)$ with weight three.
A single-crossing profile with maximum number of voters has three voters.
For instance, the profile with voters~$v_1, v_2, v_3$.

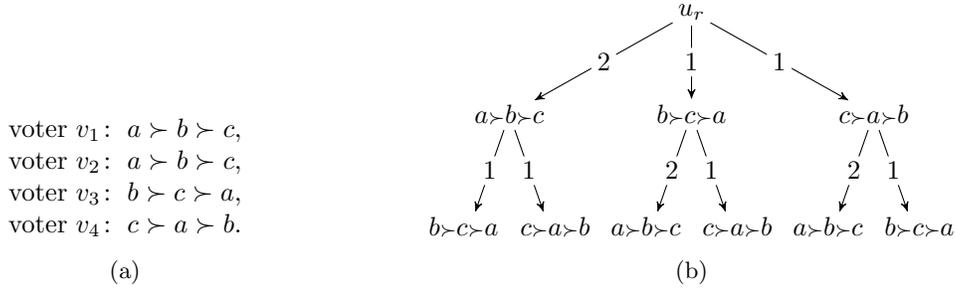
\begin{figure}[t!]
  \centering
  \begin{subfigure}[b]{.3\textwidth}
    \centering
    \begin{tabular}{@{} l @{} l @{}}
        $\defvoter v_1 \colon$\; & $a\pref b \pref c\text{,}$\\
        $\defvoter v_2 \colon$\; & $a\pref b \pref c\text{,}$\\
        $\defvoter v_3 \colon$\; & $b\pref c \pref a\text{,}$\\
        $\defvoter v_4 \colon$\; & $c\pref a \pref b\text{.}$
      \end{tabular}
      \caption{}
      \label{fig:sc-vd-example-sc-profile}
  \end{subfigure}
  \tikzset{
    treenode/.style = {align=center, text centered},
    arn/.style = {treenode,  inner sep = 1pt,  fill=white},
    arn_label/.style = {treenode, draw=white, fill=white, inner sep=2pt, midway},
    arrow/.style = {->,shorten <= 1pt, shorten >= 1pt}
  }
  \tikzstyle{level 1}=[level distance=9ex, sibling distance=16ex]
  \tikzstyle{level 2}=[level distance=10ex, sibling distance=7ex]
  \qquad
  \begin{subfigure}[b]{.6\textwidth}
    \centering
  \begin{tikzpicture}[>=stealth']
    \node[arn] (root) {$u_r$}
    child
    {
      node[arn] {$a\spref b \spref c$}
      child
      {
        node[arn] (leftmost) {$b \spref c \spref a$\phantom{$_4$}} 
        edge from parent[arrow] 
        node[arn_label] {$1$}
      } 
      child
      {
        node[arn] {\phantom{$_4$}$c \spref a \spref b$} 
        edge from parent[arrow] 
        node[arn_label] {$1$}
      }
      edge from parent[arrow] node[arn_label] {$2$}
    }
    child
    {
      node[arn] {$b\spref c \spref a$}
      child
      {
        node[arn] {$a \spref b \spref c$\phantom{$_4$}} 
        edge from parent[arrow] 
        node[arn_label] {$2$}
      } 
      child
      {
        node[arn] {\phantom{$_4$}$c \spref a \spref b$} 
        edge from parent[arrow] 
        node[arn_label] {$1$}
      }
      edge from parent[arrow] node[arn_label] {$1$}
    }
    child
    {
      node[arn] {$c\spref a \spref b$}
      child
      {
        node[arn] {$a \spref b \spref c$\phantom{$_4$}} 
        edge from parent[arrow] 
        node[arn_label] {$2$}
      } 
      child
      {
        node[arn] (last node) {\phantom{$_4$}$b \spref c \spref a$} 
        edge from parent[arrow] 
        node[arn_label] {$1$}
      }
      edge from parent[arrow] node[arn_label] {$1$}
    }
    ;
  \end{tikzpicture}
  \caption{}\label{fig:sc-vd-example-dag}
  \end{subfigure}
  \caption{An example illustrating how to construct a weighted directed graph for a given profile.
  (a) A profile with four voters and three alternatives. Note that the first two voters have the same preference orders and that this profile is not single-crossing.
  (b) A weighted directed graph for the left profile. 
  Note that we label each vertex with its corresponding preference order.
  The weight on an arc denotes the number of voters in the profile that have the preference order labeled in the source vertex.
  For instance, there is an arc from the root~$u_r$ to its left most ``child'' $a\spref b \spref c$ with weight $2$.
  This means that the left profile has two voters with preference order~$a \spref b \spref c$.}
\label{fig:sc-vd-example}
\end{figure}

\section{Conclusion}
\label{sec:conclusion}
In terms of computational complexity theory, this work is one of the starting
points for preference profiles which are ``close'' to being nicely structured.
We have shown that
making a profile single-crossing by deleting as few voters as possible can be solved in polynomial time.
In contrast, making a profile nicely structured by deleting at most~$k$ voters
or at most~$k$ alternatives is \NP-hard for all other considered cases.
However, we mention in passing that all these problems become tractable when $k$~is small:
All considered properties are characterized by a fixed number of forbidden substructures.
Thus, by branching over all possible voters (resp.\ alternatives) of each
forbidden substructure in the profile one obtains a fixed-parameter
algorithm~\cite{CyFoKoLoMaPiPiSa2015,DF13,FG06,Nie06} that is efficient for small distances.
One line of future research is to investigate more sophisticated and
more efficient (fixed-parameter) algorithms to
compute the distance of a profile to a nicely structured one~\cite{ElkLac14}.

A second line of research which was started by \citet{ELP13} for
single-peaked profiles is to study further distance measures such as ``the
number of pairs of alternatives to swap''.
Besides the domain restrictions studied in this paper, 
there is also a very nicely structured property in the literature, 
the so-called \euclid representation.
It models the ability to place voters and alternatives onto
a real line such that a voter prefers an alternative to another one if and
only if the first one is closer to the voter.
\euclid profiles are necessarily single-peaked and single-crossing.  
The last two properties, however, are not sufficient to characterize 
the \euclid profile~\cite{ChPrWo15, Coombs1964, ElkFalSkow2014}.
In fact, \citet{ChPrWo15} show that
the \euclid profile \emph{cannot} be characterized by finitely many forbidden substructures.
Nevertheless, recognizing \euclid profiles can be done in polynomial time~\cite{DoiFal1994, ElkFal2014, Knoblauch2010}.
The computational complexity of making a profile \euclid using
a minimum number of modifications remains unexplored.

A third line of research investigates whether and in which way tractability
for nicely structured preference profiles transfers to profiles that are only close
to being nicely structured (see also key questions~6 and~7 in~\cite{BCFGNW2014}).
It was started by \citet{CGS12,CGS13,FaHeHe2014,SkoYuFalElk2015,YanGuo14}
who look into several notions of nearly nicely structured profiles which are different from, but related to ours.
There are cases where the computational tractability of voting problems on nicely structured profiles
transfers to nearly nicely structured profiles and
cases where the vulnerability disappears even if the preference profile is extremely close to being nicely-structured.

\paragraph{Acknowledgements.}
Robert Bredereck was supported by the German Research Foundation (DFG), research project PAWS, NI~369/10.
Jiehua Chen was supported by the Studienstiftung des Deutschen Volkes.
Gerhard J. Woeginger was supported by the Alexander von Humboldt Foundation.


\begin{thebibliography}{53}
\providecommand{\natexlab}[1]{#1}
\providecommand{\url}[1]{\texttt{#1}}
\expandafter\ifx\csname urlstyle\endcsname\relax
  \providecommand{\doi}[1]{doi: #1}\else
  \providecommand{\doi}{doi: \begingroup \urlstyle{rm}\Url}\fi

\bibitem[Abello and Johnson(1984)]{AbeJoh1984}
J.~M. Abello and C.~R. Johnson.
\newblock How large are transitive simple majority domains?
\newblock \emph{SIAM Journal on Algebraic and Discrete Methods}, 5\penalty0
  (4):\penalty0 603--618, 1984.

\bibitem[\'{A}d\'{a}m Galambos and Reiner(2008)]{GalRei2008}
\'{A}d\'{a}m Galambos and V.~Reiner.
\newblock Acyclic sets of linear orders via the {Bruhat} orders.
\newblock \emph{Social Choice and Welfare}, 30\penalty0 (2):\penalty0 245--264,
  2008.

\bibitem[Arrow(1950)]{Arrow1950}
K.~J. Arrow.
\newblock A difficulty in the concept of social welfare.
\newblock \emph{Journal of Political Economy}, 58\penalty0 (4):\penalty0
  328--346, 1950.

\bibitem[Baigent(1987)]{Baigent1987}
N.~Baigent.
\newblock Metric rationalisation of social choice functions according to
  principles of social choice.
\newblock \emph{Mathematical Social Sciences}, 13\penalty0 (1):\penalty0
  59--65, 1987.

\bibitem[Ballester and Haeringer(2011)]{BH11}
M.~Ballester and G.~Haeringer.
\newblock A characterization of the single-peaked domain.
\newblock \emph{Social Choice and Welfare}, 36\penalty0 (2):\penalty0 305--322,
  2011.

\bibitem[Barber{\`a} et~al.(1993)Barber{\`a}, Gul, and Ennio]{BarGulSta1993}
S.~Barber{\`a}, F.~Gul, and S.~Ennio.
\newblock Generalized median voter schemes and committees.
\newblock \emph{Journal of Economic Theory}, 61\penalty0 (2):\penalty0
  262--289, 1993.

\bibitem[{Bartholdi~III} and Trick(1986)]{BarTri1986}
J.~{Bartholdi~III} and M.~A. Trick.
\newblock Stable matching with preferences derived from a psychological model.
\newblock \emph{{Operations Research Letters}}, 5\penalty0 (4):\penalty0
  165--169, 1986.

\bibitem[{Bartholdi~III} et~al.(1989){Bartholdi~III}, Tovey, and
  Trick]{BaToTr1989}
J.~J. {Bartholdi~III}, C.~A. Tovey, and M.~A. Trick.
\newblock Voting schemes for which it can be difficult to tell who won the
  election.
\newblock \emph{Social Choice and Welfare}, 6\penalty0 (2):\penalty0 157--165,
  1989.

\bibitem[Betzler et~al.(2013)Betzler, Slinko, and Uhlmann]{BetArkJoh2013}
N.~Betzler, A.~Slinko, and J.~Uhlmann.
\newblock On the computation of fully proportional representation.
\newblock \emph{Journal of Artificial Intelligence Research}, 47\penalty0
  (1):\penalty0 475--519, 2013.

\bibitem[Black(1948)]{Black1948}
D.~Black.
\newblock On the rationale of group decision making.
\newblock \emph{Journal of Political Economy}, 56\penalty0 (1):\penalty0
  23--34, 1948.

\bibitem[Brandt et~al.(2015)Brandt, Brill, Hemaspaandra, and
  Hemaspaandra]{BraBriHemHem2015}
F.~Brandt, M.~Brill, E.~Hemaspaandra, and L.~A. Hemaspaandra.
\newblock Bypassing combinatorial protections: {P}olynomial-time algorithms for
  single-peaked electorates.
\newblock \emph{Journal of Artificial Intelligence Research}, 53:\penalty0
  439--496, 2015.

\bibitem[Bredereck et~al.(2013{\natexlab{a}})Bredereck, Chen, and
  Woeginger]{BCW13}
R.~Bredereck, J.~Chen, and G.~J. Woeginger.
\newblock A characterization of the single-crossing domain.
\newblock \emph{Social Choice and Welfare}, 41\penalty0 (4):\penalty0 989--998,
  2013{\natexlab{a}}.

\bibitem[Bredereck et~al.(2013{\natexlab{b}})Bredereck, Chen, and
  Woeginger]{BreCheWoe2013}
R.~Bredereck, J.~Chen, and G.~J. Woeginger.
\newblock Are there any nicely structured preference profiles nearby?
\newblock In \emph{Proceedings of the 23rd International Joint Conference on
  Artificial Intelligence (IJCAI~'13)}, pages 62--68. AAAI Press,
  2013{\natexlab{b}}.

\bibitem[Bredereck et~al.(2014)Bredereck, Chen, Faliszewski, Guo, Niedermeier,
  and Woeginger]{BCFGNW2014}
R.~Bredereck, J.~Chen, P.~Faliszewski, J.~Guo, R.~Niedermeier, and G.~J.
  Woeginger.
\newblock Parameterized algorithmics for computational social choice: {N}ine
  research challenges.
\newblock \emph{Tsinghua Science and Technology}, 19\penalty0 (4):\penalty0
  358--373, 2014.

\bibitem[Chen et~al.(2014)Chen, Faliszewski, Niedermeier, and
  Talmon]{CheFalNieTal2014}
J.~Chen, P.~Faliszewski, R.~Niedermeier, and N.~Talmon.
\newblock Combinatorial voter control in elections.
\newblock In \emph{Proceedings of the 39th International Symposium on
  Mathematical Foundations of Computer Science (MFCS~'14)}, volume 8635 of
  \emph{Lecture Notes in Computer Science}, pages 153--164. Springer, 2014.

\bibitem[Chen et~al.(2015)Chen, Pruhs, and Woeginger]{ChPrWo15}
J.~Chen, K.~Pruhs, and G.~J. Woeginger.
\newblock The one-dimensional euclidean domain: {F}initely many obstructions
  are not enough, June 2015.
\newblock {\tt arXiv:1506.03838v1 [cs.GT]}.

\bibitem[Coombs(1964)]{Coombs1964}
C.~H. Coombs.
\newblock \emph{A Theory of Data}.
\newblock John Wiley and Sons, 1964.

\bibitem[Cormen et~al.(2009)Cormen, Leiserson, Rivest, and
  Stein]{CorLeiRivSte2009}
T.~H. Cormen, C.~E. Leiserson, R.~L. Rivest, and C.~Stein.
\newblock \emph{Introduction to Algorithms}.
\newblock MIT Press, 2009.

\bibitem[Cornaz et~al.(2012)Cornaz, Galand, and Spanjaard]{CGS12}
D.~Cornaz, L.~Galand, and O.~Spanjaard.
\newblock Bounded single-peaked width and proportional representation.
\newblock In \emph{Proceedings of the 20th European Conference on Artificial
  Intelligence (ECAI '12)}, pages 270--275. {IOS} Press, 2012.

\bibitem[Cornaz et~al.(2013)Cornaz, Galand, and Spanjaard]{CGS13}
D.~Cornaz, L.~Galand, and O.~Spanjaard.
\newblock {K}emeny elections with bounded single-peaked or single-crossing
  width.
\newblock In \emph{Proceedings of the 23rd International Joint Conference on
  Artificial Intelligence (IJCAI~'13)}, pages 76--82. AAAI Press, 2013.

\bibitem[Cygan et~al.(2015)Cygan, Fomin, Kowalik, Lokshtanov, Marx, Pilipczuk,
  Pilipczuk, and Saurabh]{CyFoKoLoMaPiPiSa2015}
M.~Cygan, F.~V. Fomin, L.~Kowalik, D.~Lokshtanov, D.~Marx, M.~Pilipczuk,
  M.~Pilipczuk, and S.~Saurabh.
\newblock \emph{Parameterized Algorithms}.
\newblock Springer, 2015.

\bibitem[Doignon and Falmagne(1994)]{DoiFal1994}
J.~Doignon and J.~Falmagne.
\newblock A polynomial time algorithm for unidimensional unfolding
  representations.
\newblock \emph{Journal of Algorithms}, 16\penalty0 (2):\penalty0 218--233,
  1994.

\bibitem[Downey and Fellows(2013)]{DF13}
R.~G. Downey and M.~R. Fellows.
\newblock \emph{Fundamentals of Parameterized Complexity}.
\newblock Springer, 2013.

\bibitem[Elkind and Faliszewski(2014)]{ElkFal2014}
E.~Elkind and P.~Faliszewski.
\newblock Recognizing 1-{E}uclidean preferences: An alternative approach.
\newblock In \emph{Proceedings of the 7th International Symposium on
  Algorithmic Game Theory (SAGT~'14)}, volume 8768 of \emph{Lecture Notes in
  Computer Science}, pages 146--157. Springer, 2014.

\bibitem[Elkind and Lackner(2014)]{ElkLac14}
E.~Elkind and M.~Lackner.
\newblock On detecting nearly structured preference profiles.
\newblock In \emph{Proceedings of the 28th AAAI Conference on Artificial
  Intelligence (AAAI~'14)}. AAAI Press, 2014.

\bibitem[Elkind et~al.(2012{\natexlab{a}})Elkind, Faliszewski, and
  Slinko]{EFS12}
E.~Elkind, P.~Faliszewski, and A.~M. Slinko.
\newblock Clone structures in voters' preferences.
\newblock In \emph{Proceedings of the 13th ACM Conference on Electronic
  Commerce (EC~'12)}, pages 496--513. ACM~Press, 2012{\natexlab{a}}.

\bibitem[Elkind et~al.(2012{\natexlab{b}})Elkind, Faliszewski, and
  Slinko]{ElkFalSli2012}
E.~Elkind, P.~Faliszewski, and A.~M. Slinko.
\newblock Rationalizations of {C}ondorcet-consistent rules via distances of
  hamming type.
\newblock \emph{Social Choice and Welfare}, 39\penalty0 (4):\penalty0 891--905,
  2012{\natexlab{b}}.

\bibitem[Elkind et~al.(2014)Elkind, Faliszewski, and Skowron]{ElkFalSkow2014}
E.~Elkind, P.~Faliszewski, and P.~Skowron.
\newblock A characterization of the single-peaked single-crossing domain.
\newblock In \emph{Proceedings of the 28th AAAI Conference on Artificial
  Intelligence (AAAI~'14)}, pages 654--660. AAAI Press, 2014.

\bibitem[Erd{\'e}lyi et~al.(2013)Erd{\'e}lyi, Lackner, and Pfandler]{ELP13}
G.~Erd{\'e}lyi, M.~Lackner, and A.~Pfandler.
\newblock Computational aspects of nearly single-peaked electorates.
\newblock In \emph{Proceedings of the 27th AAAI Conference on Artificial
  Intelligence (AAAI~'13)}, pages 283--289. AAAI Press, 2013.

\bibitem[Escoffier et~al.(2008)Escoffier, Lang, and
  \"{O}zt\"{u}rk]{EscLanOez2008}
B.~Escoffier, J.~Lang, and M.~\"{O}zt\"{u}rk.
\newblock {Single-Peaked Consistency and its Complexity}.
\newblock In \emph{Proceedings of the 18th European Conference on Artificial
  Intelligence (ECAI '08)}, pages 366--370. {IOS} Press, 2008.

\bibitem[Faliszewski et~al.(2011)Faliszewski, Hemaspaandra, Hemaspaandra, and
  Rothe]{FHHR2011}
P.~Faliszewski, E.~Hemaspaandra, L.~A. Hemaspaandra, and J.~Rothe.
\newblock The shield that never was: {S}ocieties with single-peaked preferences
  are more open to manipulation and control.
\newblock \emph{Information and Computation}, 209\penalty0 (2):\penalty0
  89--107, 2011.

\bibitem[Faliszewski et~al.(2014)Faliszewski, Hemaspaandra, and
  Hemaspaandra]{FaHeHe2014}
P.~Faliszewski, E.~Hemaspaandra, and L.~A. Hemaspaandra.
\newblock The complexity of manipulative attacks in nearly single-peaked
  electorates.
\newblock \emph{Artificial Intelligence}, 207:\penalty0 69--99, 2014.

\bibitem[Flum and Grohe(2006)]{FG06}
J.~Flum and M.~Grohe.
\newblock \emph{Parameterized Complexity Theory}.
\newblock Springer, 2006.

\bibitem[Garey and Johnson(1979)]{GJ79}
M.~R. Garey and D.~S. Johnson.
\newblock \emph{Computers and Intractability---{A} Guide to the Theory of
  {NP}-Completeness}.
\newblock W. H. Freeman and Company, 1979.

\bibitem[Inada(1964)]{Inada1964}
K.~Inada.
\newblock A note on the simple majority decision rule.
\newblock \emph{Econometrica}, 32\penalty0 (32):\penalty0 525--531, 1964.

\bibitem[Inada(1969)]{Inada1969}
K.~Inada.
\newblock The simple majority decision rule.
\newblock \emph{Econometrica}, 37\penalty0 (3):\penalty0 490--506, 1969.

\bibitem[Klaus et~al.(1997)Klaus, Peters, and Storcken]{KlPeSt1997}
B.~Klaus, H.~Peters, and T.~Storcken.
\newblock Strategy-proof division of a private good when preferences are
  single-dipped.
\newblock \emph{Economics Letters}, 55\penalty0 (3):\penalty0 339--346, 1997.

\bibitem[Knoblauch(2010)]{Knoblauch2010}
V.~Knoblauch.
\newblock Recognizing one-dimensional {E}uclidean preference profiles.
\newblock \emph{Journal of Mathematical Economics}, 46\penalty0 (1):\penalty0
  1--5, 2010.

\bibitem[Meskanen and Nurmi(2008)]{MesNur2008}
T.~Meskanen and H.~Nurmi.
\newblock Closeness counts in social choice.
\newblock In \emph{Power, Freedom, and Voting}, pages 289--306. Springer, 2008.

\bibitem[Mirrlees(1971)]{Mirrlees1971}
J.~A. Mirrlees.
\newblock An exploration in the theory of optimal income taxation.
\newblock \emph{Review of Economic Studies}, 38\penalty0 (2):\penalty0
  175--208, 1971.

\bibitem[Monjardet(2009)]{Monjardet2009}
B.~Monjardet.
\newblock Acyclic domains of linear orders: {A} survey.
\newblock In \emph{The Mathematics of Preference, Choice and Order}, pages
  139--160. Springer, 2009.

\bibitem[Moulin(1980)]{Moulin1980}
H.~Moulin.
\newblock On strategy-proofness and single peakedness.
\newblock \emph{Public Choice}, 35\penalty0 (4):\penalty0 437--455, 1980.

\bibitem[Niedermeier(2006)]{Nie06}
R.~Niedermeier.
\newblock \emph{Invitation to Fixed-Parameter Algorithms}.
\newblock Oxford University Press, 2006.

\bibitem[Puppe and Slinko(2015)]{PupSli2015}
C.~Puppe and A.~M. Slinko.
\newblock Condorcet domains, median graphs and the single crossing property,
  May 2015.
\newblock {\tt arXiv:1505.06982v1 [cs.GT]}.

\bibitem[Roberts(1977)]{Roberts1977}
K.~W. Roberts.
\newblock Voting over income tax schedules.
\newblock \emph{Journal of Public Economics}, 8\penalty0 (3):\penalty0
  329--340, 1977.

\bibitem[Sen and Pattanaik(1969)]{SePa1969}
A.~Sen and P.~K. Pattanaik.
\newblock Necessary and sufficient conditions for rational choice under
  majority decision.
\newblock \emph{Journal of Economic Theory}, 1\penalty0 (2):\penalty0 178--202,
  1969.

\bibitem[Sen(1966)]{Sen1966}
A.~K. Sen.
\newblock A possibility theorem on majority decisions.
\newblock \emph{Econometrica}, 34\penalty0 (2):\penalty0 491--499, 1966.

\bibitem[Skowron et~al.(2015)Skowron, Yu, Faliszewski, and
  Elkind]{SkoYuFalElk2015}
P.~Skowron, L.~Yu, P.~Faliszewski, and E.~Elkind.
\newblock The complexity of fully proportional representation for
  single-crossing electorates.
\newblock \emph{Theoretical Computer Science}, 569:\penalty0 43--57, 2015.

\bibitem[Sui et~al.(2013)Sui, Francois{-}Nienaber, and
  Boutilier]{SuiFraNieBou2013}
X.~Sui, A.~Francois{-}Nienaber, and C.~Boutilier.
\newblock Multi-dimensional single-peaked consistency and its approximations.
\newblock In \emph{Proceedings of the 23rd International Joint Conference on
  Artificial Intelligence (IJCAI~'13)}. AAAI Press, 2013.

\bibitem[Tideman(1987)]{Tid1987}
T.~N. Tideman.
\newblock Independence of clones as a criterion for voting rules.
\newblock \emph{Social Choice and Welfare}, 4\penalty0 (3):\penalty0 185--206,
  1987.

\bibitem[Walsh(2007)]{Walsh2007}
T.~Walsh.
\newblock Uncertainty in preference elicitation and aggregation.
\newblock In \emph{Proceedings of the 22nd Conference on Artificial
  Intelligence (AAAI~'07)}, pages 3--8. AAAI Press, 2007.

\bibitem[Yang and Guo(2014)]{YanGuo14}
Y.~Yang and J.~Guo.
\newblock The control complexity of $r$-approval: from the single-peaked case
  to the general case.
\newblock In \emph{Proceedings of the 13th International Conference on
  Autonomous Agents and Multiagent Systems (AAMAS~'14)}, pages 621--628.
  IFAAMAS, 2014.

\bibitem[Young(1977)]{You77}
H.~P. Young.
\newblock Extending {C}ondorcet's rule.
\newblock \emph{Journal of Economic Theory}, 16\penalty0 (2):\penalty0
  335--353, 1977.

\end{thebibliography}
\end{document}